\documentclass[journal]{IEEEtran}

\usepackage{graphics} 
\usepackage{epsfig} 
\usepackage{mathptmx} 
\usepackage{times} 
\usepackage{amsthm}
\usepackage{amsmath} 
\usepackage{amssymb}  
\usepackage{cuted}
\usepackage{float}
\usepackage{color}
\usepackage{url}

\usepackage[ruled,vlined,linesnumbered,noresetcount]{algorithm2e}
\usepackage{tabularx}

\usepackage{caption,subcaption,cite}

\newtheorem{theorem}{Theorem}
\newtheorem{corollary}{Corollary}

\newtheorem{proposition}{Proposition}
\newtheorem{definition}{Definition}
\newtheorem{remark}{Remark}
\newtheorem{example}{Example}

\DeclareMathOperator{\Di}{Diag}

\makeatletter
\def\ps@IEEEtitlepagestyle{%
  \def\@oddfoot{\mycopyrightnotice}%
  \def\@evenfoot{}%
}
\def\mycopyrightnotice{%
  \begin{minipage}{\textwidth}
  \centering \scriptsize
  Copyright~\copyright~2021 IEEE. Personal use of this material is permitted. Permission from IEEE must be obtained for all other uses, in any current or future media, including reprinting/republishing this material for advertising or promotional purposes,creating new collective works, for resale or redistribution to servers or lists, or reuse of any copyrighted component of this work in other works. 
  \end{minipage}
}
\makeatother

\begin{document}
%
\title{
A Dynamic Game Framework for Rational and Persistent Robot Deception with an Application to Deceptive Pursuit-Evasion
}
%
%
%

\author{Linan~Huang,~\IEEEmembership{Student Member,~IEEE,}
        Quanyan~Zhu,~\IEEEmembership{Member,~IEEE}
   \thanks{This  paper  has  been  accepted  for  publication  in IEEE Transactions on Automation Science and Engineering}
\thanks{
This research is partially supported by awards ECCS-1847056, CNS-1544782, CNS-2027884, and SES-1541164 from National Science of Foundation (NSF), and grant W911NF-19-1-0041 from Army Research Office (ARO). 
}      

\thanks{L. Huang and Q. Zhu are with the Department of Electrical and Computer Engineering,  New York University,  370 Jay Street, Brooklyn, NY 11201, USA; E-mail: \{lh2328,qz494\}@nyu.edu}

\thanks{Digital Object Identifier 10.1109/TASE.2021.3097286}

}
\maketitle

\begin{abstract}
This paper studies rational and persistent deception among intelligent robots to enhance security and operational efficiency. 
We present an N-player K-stage game with an asymmetric information structure where each robot's private information is modeled as a random variable or its type. 
The deception is persistent as each robot's private type remains unknown to other robots for all stages. 
 The deception is rational as robots aim to achieve their deception goals at minimum cost.  
Each robot forms a {dynamic belief of others' types  based on intrinsic or extrinsic information.}  
 Perfect Bayesian Nash Equilibrium (PBNE) is a natural solution concept for dynamic games of incomplete information. 
 {Due to its requirements of sequential rationality and belief consistency, PBNE} provides a reliable prediction of players' actions, beliefs, and expected cumulative costs over the entire K stages. 
 The contribution of this work is fourfold. 
 First, we identify the PBNE computation as a nonlinear stochastic control problem and characterize the structures of players' actions and costs under PBNE. 
 We further derive a set of extended Riccati equations with \textcolor{black}{cognitive coupling}
  under the linear-quadratic setting and extrinsic belief dynamics. 
 Second, we develop a receding-horizon algorithm with low temporal and spatial complexity to compute PBNE under intrinsic belief dynamics. 
 Third, we investigate a deceptive pursuit-evasion game as a case study and use numerical experiments to corroborate the results.    
 Finally, we propose metrics, such as deceivability, reachability, and the price of deception, to evaluate the strategy design and the system performance under deception. 
 
\end{abstract}

\renewcommand{\abstractname}{Note to Practitioners}
\begin{abstract}
Recent advances in automation and adaptive control in multi-agent systems enable robots to use deception to accomplish their objectives. 
Deception involves intentional information hiding to compromise the security and operational efficiency of the robotic systems. 
This work proposes a dynamic game framework to quantify the impact of deception, understand the robots' behaviors and intentions, and design cost-efficient strategies under the deception that persists over stages. Existing researches on robot deception have relied on experiments while this work aims to lay a theoretical foundation of deception
{with quantitative metrics,} such as deceivability and the price of deception. 
The proposed model has wide applications, including cooperative robots, pursuit and evasion, and human-robot teaming. 
The pursuit-evasion games are used as case studies to show how the deceiver can amplify the deception by belief manipulation and how the deceived robots can reduce the negative impact of deception by enhanced maneuverability and Bayesian learning.
The future work would focus on \textcolor{black}{designing}
 cooperative deception among swarm robotics and  robotic systems that are robust to or further benefit from deception.  
\end{abstract}


%
\IEEEpeerreviewmaketitle

\section{Introduction}
%
%
%
%
\IEEEPARstart{D}{eception}  is a ubiquitous phenomenon in biology \cite{smith2004we}, military \cite{howard1995strategic}, politics and media \cite{cowen2017network}, and cyberspace \cite{al2019dynamic}. 
In particular, deception plays an increasingly significant role in cyber-physical systems, including autonomous vehicles and robots driven by artificial intelligence (AI). 
Recent advances in these AI-enabled technologies have not only allowed robots to adapt to the dynamic environment via real-time observations, but also made them deceivable.  
A deceiver can intentionally hide or reveal selected information to alter the \textcolor{black}{beliefs and behaviors} of the target robots for a higher reward. 
Since deception has many forms and delivery methods, understanding deception in a unified and quantitative framework is an indispensable step toward assessing the outcomes, measuring the impact, and designing strategies. 
This work aims to design  robots 
that can interact with others efficiently under deceptive environments. 

We identify the following challenges and features of robot deception. First, by definition, deception involves at least two participants interacting with each other. 
An intelligent robot should further consider other participants' rationality, predict their potential deceptive behaviors, and adjust its actions accordingly to alleviate the negative effect of deception. 
Second, due to the robots' dynamic nature, one-shot deception can exert a subsequent influence. The participating robots need to form long-term objectives to deceive or counter-deceive other robots. 
The multi-stage interactions also make it possible for the deceiver to apply deception at different stages. 
Third, each robot contains heterogeneous private information, which results in an asymmetric cognition structure; i.e., robots can form different beliefs over the same piece of unknown information. 
Thus, besides the couplings of state dynamics and costs, the multi-agent system further has \textit{cognitive coupling}; i.e., each robot's behaviors are not only affected by its own belief but also the beliefs of the others. 

To capture these features, we model the deceptive interaction between $N$ strategic robots as a dynamic game of incomplete information. 
During the finite $K$ stages of interaction, $N$ robots accomplish non-cooperative tasks such as  pursuit-evasion in the battlefield \cite{li2011defending}  or cooperative tasks such as collective towing  \cite{sreenath2013dynamics}. 
Robots introduce deception in the above interacting scenarios due to antagonism, selfishness, and privacy concerns. 
Following Harsanyi's approach \cite{harsanyi1967games}, we capture each robot's private information by a random variable. 
The realization of the  random variable, which is called the robot's \textit{type}, is known only to itself, while the support of the  random variable, which contains all its possible types, is known to all robots. 
Take the pursuit-evasion scenario as an example, due to the constraints of weather, terrain, and weapon, both the evading and the pursuing robots know the feasible beachheads for the evader to land on. 
However, the evader chooses only one beachhead {as his true target and the evader's choice, i.e., his type, is unknown to the pursuer. 
The pursuer in the battlefield \textcolor{black}{knows} the existence of the deception and \textcolor{black}{learns} to counter the deception by forming and updating her belief based on real-time observations.}  
Since these tasks are usually time-constrained, robots cannot wait and freeze until they have learned the true type. Instead, they have to  take concurrent actions while the deceiver's type remains uncertain. 

{
We consider two classes of belief dynamics based on whether robots exploit the intrinsic information such as the prediction of other robots' actions, or the extrinsic information to update their beliefs.} 
Each robot aims to minimize its expected cumulative cost over $K$ stages. 
Since the expectation involves its {
$K$-stage belief sequence of other players' private types, its actions should be sequentially rational under its belief sequence and the belief sequence should be consistent with the belief dynamics \textcolor{black}{as well}. 
These two requirements lead to the solution concept of Perfect Bayesian Nash Equilibrium (PBNE) where a player's unilateral deviation from the equilibrium increases his long-run cost. 
By appending the belief state (i.e., all players' beliefs under all possible types) to the system state, the PBNE computation is equivalent to a multi-agent nonlinear stochastic control problem and the method of dynamic programming applies. Without loss of generality, we characterize the structure of the action and the cost under PBNE as a feedback function of the belief state and the system state at the current stage.  
To provide an offline evaluation metric of the equilibrium cost under incomplete information, we use the expected equilibrium cost under complete information as a benchmark and define the \textit{Price of Deception} (PoD).}  

{
Due to their tractability and generality, we focus on  \textcolor{black}{incomplete-information} Linear-Quadratic (LQ) games with \textit{extrinsic belief dynamics} to obtain the PBNE action that is unique and affine to the system state.}  
We obtain a set of extended Riccati equations, which explicitly characterizes the coupling in the state dynamics, costs, and cognition of all robots. 
Under proper decoupling structures, the extended Riccati equations degenerate to the classical Riccati equations for the problems of optimal control \textcolor{black}{or  complete-information LQ games}. 
Under the  \textcolor{black}{incomplete-information} LQ \textcolor{black}{games} with \textit{intrinsic belief dynamics}, the equilibrium action is in general not affine feedback of the system state. 
Thus, we adopt a receding-horizon approach to provide a reasonable approximation of PBNE; i.e.,  instead of offline planning of all $K$-stage actions before the game starts, players recompute their actions based on the real-time observations and their updated beliefs at each new stage during the interaction. 

Finally, we investigate a target protection problem where an evader aims to deceptively reach one of the possible targets and simultaneously evade the pursuer. The game has doubled-sided asymmetric information.
The evader's private or hidden information is his true target while the pursuer's private information is her capability to maneuver or the \textit{maneuverability}. 
We propose multi-dimensional metrics, including the \textit{stage of truth revelation}  and the \textit{endpoint distance}, to assess the deception impact.   
We define the concept of \textit{deceivability} to characterize the fundamental limits of deception and investigate how it is affected by the  \textit{distinguishability} of the private information.  
We compare the proposed control policy with two heuristic polices to demonstrate its efficacy to counter deception at a much lower cost.  
We show that Bayesian learning can significantly reduce the impact of initial belief manipulation and result in a win-win situation for some cases.  
The increase of the pursuer's maneuverability  improves her control performance under deception yet has a marginal effect. 
We also find that applying deception to counter deception is not always effective; e.g., it can be beneficial for a \textcolor{black}{less maneuverable} pursuer to disguise as a \textcolor{black}{more maneuverable} pursuer but not vice versa.  
The numerical results corroborate that PoD can exceed $1$; i.e., deception among players may not only benefit the deceiver but also the deceivee. 

\subsection{Related Works}
The secure and efficient operation of robots, autonomous vehicles, and industrial control systems is vital for recent advances in technologies. 
Many works \cite{thing2016autonomous,huangdynamic,zhao2020finite} have investigated how to  protect these systems  from \textcolor{black}{various} attacks on \textcolor{black}{sensor} measurements \cite{9177270}, communication channels  \cite{bhattacharya2010game}, and control signals  \cite{zhang2017stealthy,8709954}. 
Deception is a key feature of sophisticated attacks with a focus on intentionally hiding private information \cite{ayub2019adaptive,HuangAPT}, introducing randomness \cite{9029607}, and  manipulating other players' beliefs \cite{ornik2018deception,horak2017manipulating}. 
Deception in robotic systems can be conducted  through visual displays \cite{brewer2006visual}, 
facial expressions and body gestures \cite{shim2014other}, and trajectories \cite{dragan2014analysis,ayub2019adaptive}. 
Existing works on robot deception are largely based on experimental approaches \cite{ayub2019adaptive,shim2013taxonomy,wagner2011acting}. There is a need for a formal and quantitative framework to assess the deception impact, understand the fundamental limit and \textcolor{black}{tradeoff} of deception, and determine real-time strategies. 
Compared to the theoretical works of deceptive path planning and goal recognition \cite{Pmaster2019,e22010088}, which focus on identifying the \textcolor{black}{true target} behind deception, our work further determines optimal and cost-effective \textcolor{black}{control policies} to counteract deception \textcolor{black}{and physically protect the true target}; e.g.,  \textcolor{black}{the pursuer adopts the action sequence of minimum cost to reach and protect the true beachhead selected by the evader.}  
\textcolor{black}{Compared to control-theoretic deception frameworks based on Markov decision processes \cite{9029607,ornik2018deception} and stochastic games \cite{Deceptionwill}, we adopt a state-space representation to better characterize the physical dynamics of robots and autonomous vehicles.} 

Game models such as hypergames \cite{bennett1980hypergames}, dynamic Bayesian games \cite{HuangAPT}, partially observable stochastic games \cite{horak2017manipulating,He2018}, and games that involve signaling mechanisms \cite{pawlick2015deception,huang2020gameV2} 
 have been adopted as natural analytic paradigms to understand deception between  intelligent players. 
\textcolor{black}{The} computation of  equilibrium solutions for dynamic games of incomplete information, especially ones with non-classical information structure \cite{sandell1974solution}, is often a challenging task. 
 Previous works have adopted conjugate prior assumptions to simplify Bayesian update and decouple the forward type estimation and backward action optimization under a finite state space and a continuous type space \cite{huang2018analysis,huang2019adaptive}. To solve the coupling between players' belief dynamics and the multi-agent optimal control problem in the context of robotic systems where states are continuous and constrained by physical dynamics with noises, we adopt a receding-horizon approach to compute PBNE, which yields computationally tractable online strategies for the players. Similar receding-horizon approaches have been used in other contexts, including cyber-physical systems \cite{8704275}, military air operation \cite{cruz2002moving}, and autonomous racing \cite{liniger2019noncooperative}. 

\subsection{Notations and Organization of the Paper}
Calligraphic letter $\mathcal{A}$ defines a  set and $|\mathcal{A}|$ represents its cardinality. Define $\mathcal{B} \setminus \mathcal{A}$ as the set of elements in $\mathcal{B} $ but not in $\mathcal{A}$. 
\textcolor{black}{The Euclidean norm of a vector $x$ is represented by $||x||_2$.} 
Let $\mathbb{E}_{a\sim A}[f(a)]$ denote the expectation of $f(a)$ over random variable $a$ whose probability distribution is $A$. 
Let $'$ represent matrix transpose and $\Di[a_1,\cdots,a_N]$ represent \textcolor{black}{a} block diagonal matrix with \textcolor{black}{possibly non-square matrices} $a_i, i\in \mathcal{N}$, \textcolor{black}{on its diagonal}. 
Define $\{a_i\}_{i\in \mathcal{N}}:=\{a_1,\cdots,a_N\}$ as \textcolor{black}{a} set of $N$ elements,  
$[a_i]_{i\in \mathcal{N}}:=[a_1,\cdots,a_N]$ as $N$ \textcolor{black}{block matrices of the same number of rows} arranged in one row vector, and $[a_1;\cdots;a_N]=[a_1,\cdots,a_N]'$ as $N$ \textcolor{black}{block matrices of the same number of columns}  arranged in one column vector. 
Let $\mathbf{I}_{r}, \mathbf{0}_{m,n}  $ be \textcolor{black}{the} $r \times r$ identity matrix and \textcolor{black}{the} $m\times n$ zero matrix, respectively. 
The superscript $k\in \mathcal{K}$ is the stage index and the subscript $i\in \mathcal{N}$ is the player index. 
We omit a function's arguments when there is no  ambiguity, e.g., $S_i^k :=S_i^k({\beta}_i^{k},\theta_i)$. 
A piece of information for a group of players is called \textit{common knowledge} if all players know it, all players know that all players know it, and so on ad infinitum. We summarize main notations in Table \ref{table:notation}. 

The rest of paper is organized as follows. 
Section \ref{sec: generalmodel} introduces the dynamic game of incomplete information and the solution concept of PBNE. 
To obtain explicit and practical solutions, we consider a class of a linear-quadratic problems in Section \ref{sec:LQ} and obtain a set of  extended Riccati equations.  
We present a case study of deceptive pursuit-evasion in Section \ref{sec:case study} and  Section \ref{sec:conclusion} concludes the paper.

\begin{table}[t]
\centering
\caption{Summary of variables and their meanings. 
\label{table:notation}}
\begin{tabularx}{\columnwidth}{l X} 
     \hline
\textbf{Variable} & \textbf{Meaning} \\ \hline
$\mathcal{N}:=\{1,2,\cdots,N\}$       & Set of $N$ players in the dynamic game \\
$\mathcal{K}:=\{0,1,2,\cdots,K\}$      & Set of $K$ discrete stages in the dynamic game     \\
$\Theta_i:=\{\theta_i^1, \theta_i^2, \cdots,\theta_i^{N_i}\}$   & Set of $N_i$ possible types for player   $i\in \mathcal{N}$  \\
$\theta_i\in \Theta_i$     & Type of player   $i\in \mathcal{N}$    \\
$\theta:=[\theta_1,\cdots,\theta_N]$      & $N$ players' joint type  \\ 
$\Theta_{-i}:=\prod_{j\in \mathcal{N}\setminus \{i\}}\Theta_j$ & Set of types of all players except for player $i$ \\
$\theta_{-i}:=[\theta_j]_{j\in \mathcal{N}\setminus \{i\}}\in \Theta_{-i}$ & Types of all players except for player $i$\\
$\Delta (\Theta_{-i})$     & Set of probability distributions over set $\Theta_{-i}$  \\
$\Xi_{i}(\cdot)$  &  Probability distribution of player $i$'s type\\
$\Xi=[\Xi_{i}]_{i\in \mathcal{N}}$  &  Probability distribution of the joint type $\theta$ \\
$\Xi_{w}(\cdot)$  &  Probability distribution of noise $w^k, \forall k\in \mathcal{K}$\\
$x^k\in \mathbb{R}^{n\times 1}$      & System state of dimension ${n}$  at stage $k$  \\ 
$x_i^k\in \mathbb{R}^{n_i\times 1}$    & Player $i$'s state of dimension $n_i$ at stage $k$ \\
\textcolor{black}{$[\hat{x}_{i}^k(\theta_i) ]_{k\in \mathcal{K}}$} &  \textcolor{black}{Reference trajectory for player $i$ of type $\theta_i$}\\
\textcolor{black}{$\beta^k_i\in \Lambda_i \subseteq [0,1]^{|\Theta_{-i}| \times |\Theta_i|}$}      & \textcolor{black}{Player $i$'s belief state at stage $k$} \\
\textcolor{black}{$\beta^k=[\beta_i^k]_{i\in \mathcal{N}}\in \Lambda$ }     & \textcolor{black}{$N$ players' joint belief state at stage $k$}\\
$h^k:=[x^0,\cdots,x^k]\in \mathcal{H}^k$     & State history     \\
$f^k$ & State transition function at stage $k$ \\
$\Gamma^k_i$ & Player $i$'s belief transition function at stage $k$ \\
\textcolor{black}{$g^k_i$} & \textcolor{black}{Player $i$'s cost function at stage $k$} \\
\textcolor{black}{$V_i^k(\beta^k,x^k,\theta_i)$} & \textcolor{black}{Player $i$'s PBNE cost} \\
\textcolor{black}{$\bar{V}_i^k(x^k,\theta)$} & \textcolor{black}{Player $i$'s PBNE cost when all players' types are \textit{common knowledge}} \\
$u_i^k\in \mathbb{R}^{m_i\times 1}$      &  Player $i$'s action of dimension $m_i$  at stage $k$   \\
$u^k:=[u_1^k,\cdots,u_N^k]$      & $N$ players' joint action at stage $k$  \\
$u_i^{k_0:K}:=[u_i^{k_0}, \cdots, u_i^{K}]$
 & Player $i$'s action sequence from $k_0$ to $K$\\
 $u^{k_0:K}:=[u_i^{k_0:K},u_{-i}^{k_0:K}]$
 & Player $i$'s and all other players' control sequences from stage $k_0$ to $K$\\
$l_i^k(\theta_{-i}|h^k,\theta_i)$      & Player $i$'s belief  at stage $k$, i.e., the probability of other players' types  being $\theta_{-i}$ based on player $i$'s available information of $h^k,\theta_i$ \\
\hline
\end{tabularx}
\end{table}


\section{ Dynamic Game with Private Types }
\label{sec: generalmodel}
{We model deception as a $K$-stage game consisting of $N$ robots as players and \textcolor{black}{each robot has} asymmetric information.} 
Let $\mathcal{N}:=\{1, \cdots, N\}$ be the set of $N$ players and $\mathcal{K}:=\{0,1,2,\allowbreak
\cdots,K\}$  be the set of $K$ discrete stages. Private information of player $i\in \mathcal{N}$, i.e., \textcolor{black}{his} type $\theta_i$, is modeled as the realization of a discrete random variable with a finite support $\Theta_i:=\{\theta_i^1, \theta_i^2, \allowbreak
\cdots,\theta_i^{N_i}\}$ and a prior probability distribution $\Xi_{i}(\cdot)$. Hence, $N_i$ is the number of possible types for player $i$ {and $\Xi_{i}(\theta_i)$ is the probability that player $i$'s type is $\theta_i$. Define shorthand notation $\Xi:=[\Xi_i]_{i\in \mathcal{N}}$ and let $\Theta_{-i}:=\prod_{j\in \mathcal{N}\setminus \{i\}}\Theta_j$ be the set of types of all players except for player $i\in\mathcal{N}$.}   
Each player $i$ knows the value of his own type $\theta_i$, but does not know the values of other players' types $\theta_{-i}:=[\theta_j]_{ j\in \mathcal{N}\setminus \{i\}}\in \Theta_{-i}$, throughout $K$ stages of the game. 
The system state dynamics under $N$ players' joint action $u^k:=[u_1^k,\cdots,u_N^k]$, joint type $\theta:=[\theta_1,\cdots,\theta_N]$,
and an additive \textcolor{black}{external} noise $w^k\in \mathbb{R}^{n \times 1}$ are shown in \eqref{eq:state dynamics}:  
\begin{align}
\label{eq:state dynamics}
\begin{split}
 x^{k+1} =f^k(x^k,u^k_1,\cdots,u^k_N, \theta_1,\cdots,\theta_N)+w^k, k\in \mathcal{K}\setminus \{K\}. 
 \end{split}
\end{align}
The dynamics in (\ref{eq:state dynamics}) can have different interpretations based on applications. 
In the pursuit-evasion scenario as in \cite{li2011defending}, $x_i^k\in \mathbb{R}^{n_i\times 1}$ represents robot $i$'s local states such as its location and speed. The system state $x^k\in \mathbb{R}^{n \times 1}$ can be explicitly represented by $N$ robots' joint state $[x_1^k,\cdots,x_N^k]$ with $n=\sum_{{i}=1}^N n_{{i}}$. 
In the application where $N$ robots cooperatively transport a payload, e.g.,  \cite{hajieghrary2017cooperative,sreenath2013dynamics},  system state $x^k\in \mathbb{R}^{n \times 1}$ represents the payload's location and posture, which does  not explicitly relate to robots' local states. 
The noise sequence $[w^k]_{ k\in \mathcal{K}}$ assumed to  be  independent with probability density function $\Xi_{w}(\cdot)$, i.e., $\mathbb{E}_{w^k,w^h\sim \Xi_{w}}[w^k (w^h)']=0,\forall k\in \mathcal{K}, h\in \mathcal{K}\setminus \{k\}$. 
The noise is not necessarily Gaussian distributed 
but is assumed to have a zero mean, i.e., $\mathbb{E}_{w^k\sim \Xi_w}[w^k]=0, \forall k\in \mathcal{K}$.  
We assume that system dynamics \eqref{eq:state dynamics} are multi-agent controllable as defined in Definition \ref{def:Controllability} {so that players can design their deceptive actions to reach the entire state space in finite stages.}  
\begin{definition}[\textbf{Multi-Agent Controllability}]
\label{def:Controllability}
System dynamics \eqref{eq:state dynamics} are called multi-agent controllable if for any target state $x^k\in \mathbb{R}^{n \times 1}$ at stage $k\in \mathcal{K}\setminus \{0\}$, initial state $x^0\in \mathbb{R}^{n \times 1}$,  and joint type  $\theta \in \Theta$,  there exists a sequence of finite joint actions $u^{0:k}$ that drive the system state from $x^0$ to $x^k$ in expectation. 
\end{definition}

\subsection{Forward Belief Dynamics}
At each stage $k\in \mathcal{K}$, the information available to player $i$ compromises all players' state history $h^k:=[x^0,\cdots,x^k]\in \mathcal{H}^k$ as well as his own type value $\theta_i$. 
Define $\Delta (\Theta_{-i})$ as the set of probability distributions over set $\Theta_{-i}$. 
Each player $i$ at stage $k$ forms a belief $l_i^k: \mathcal{H}^k \times \Theta_i \mapsto \bigtriangleup \Theta_{-i}$ based on his available information. Thus, 
 $l_i^k(\cdot | h^k, \theta_i)$  is a  probability measure of other players' types, i.e.,  $\sum_{{\theta}_{-i}\in \Theta_{-i}} l_i^k({\theta}_{-i} | h^k, \theta_i)=1, \forall h^k\in \mathcal{H}^k, \theta_i\in \Theta_i$. 
{
Define a vector 
\begin{equation*}
    \beta_i^k:=[l_i^k(\theta_{{-i}}|h^k,\theta_i^1),l_i^k(\theta_{{-i}}|h^k,\theta_i^2),\cdots,l_i^k(\theta_{{-i}}|h^k,\theta_i^{N_i})]_{\theta_{-i}\in \Theta_{-i}}
\end{equation*}
as player $i$'s belief state at stage $k\in \mathcal{K}$. 
We assume that the set of belief states is independent of stages, i.e., $\beta_i^k\in \Lambda_i \subseteq [0,1]^{|\Theta_{-i}| \times |\Theta_i|}$. 
Then, we can represent player $i$'s belief dynamics as 
\begin{equation}
\label{eq:belief dynamics}
\beta_i^{k+1}:=\Gamma^k_i (\beta_i^{k}, u^k, w^k, \theta_i), \forall k\in \{0,\cdots,K-1\}.     
\end{equation}
Note that the belief transition function $\Gamma^k_i$ can be different for each $i$ and $k$, i.e., players' belief updates can be heterogeneous and time-varying. 
\textcolor{black}{Define $\beta^k:=[\beta_i^k]_{i\in \mathcal{N}}\in \Lambda:=\prod_{i\in \mathcal{N}}\Lambda_i$.} 
In this work, we assume that the initial beliefs of all players of all types $\beta^0$  and the belief update rules $\Gamma_i^k,\forall i\in \mathcal{N}, \forall k\in \{0,\cdots,K-1\}$, are \textit{common knowledge}. 
In the next two subsections, we provide two specific forms of $\Gamma^k_i$ that rely on \textit{intrinsic} and \textit{extrinsic} information, respectively. }  
{\subsubsection{Bayesian Belief Dynamics}
\label{subsec:bayes belief dyna}
The most common belief update rule $\Gamma^k_i$ in \eqref{eq:belief dynamics} for player $i$ at stage $k+1$ uses Bayesian inference. Given the knowledge of the sequential state observations $x^k,x^{k+1}$ and all players' actions $u^k$, each player $i$ of type $\theta_i\in \Theta_i$ at stage $k+1$ can update his belief as follows: $\forall  \theta_{-i}\in \Theta_{-i}$,}
\begin{equation}
\label{eq:Bayesian}
\begin{split}
&l_i^{k+1}(\theta_{-i} | h^{k+1},\theta_i)\\
&\quad \quad \quad \quad 
=\frac{l_i^{k}(\theta_{-i}|h^k, \theta_i)  \Pr(x^{k+1}| \theta_{-i}, x^k, \theta_i)}{\sum_{\bar{\theta}_{-i}\in \Theta_{-i}} l_i^{k}(\bar{\theta}_{-i}| h^{k},\theta_i)\Pr(x^{k+1}|\bar{\theta}_{-i},x^{k},\theta_i) }. 
\end{split}
\end{equation}
In \eqref{eq:Bayesian}, we use \textcolor{black}{the} Markov property, i.e.,  $ \Pr(x^{k+1}| \theta_{-i}, h^k, \theta_i)=\Pr(x^{k+1}| \theta_{-i}, x^k, \theta_i)=\Xi_w(x^{k+1}-f^k(x^k, u^k,\theta) )$.  
The denominator is positive as $w^k\in \mathbb{R}^{n\times 1}$. 

\begin{remark}[\textbf{Actions Reveal Type Information}]
Even if the state dynamics $f^k$ in \eqref{eq:state dynamics} are independent of $\theta_{j}, \forall j\in \mathcal{N}\setminus \{i\}$, player $i\in \mathcal{N}$ can still learn player $j$' type via \eqref{eq:Bayesian} as player $j$'s action $u^k_j$ is a function\footnote{\textcolor{black}{Each player's action is a function of his type as his cost is related to his type and the action aims to minimize his cost.}} of his type $\theta_j$. 
\end{remark}

{\subsubsection{Markov-Chain Belief Dynamics}
In section \ref{subsec:bayes belief dyna}, we assume that players can exploit the \textit{intrinsic} information of state dynamics $f^k$, state observations $x^k, x^{k+1}$, and the prediction of all players' actions $u^k$. 
Since the above \textit{intrinsic} information may not be available in practice, we consider \textcolor{black}{the} belief dynamics with \textit{extrinsic} information in this subsection. 
In particular, we assume that each player $i$'s belief dynamics $\beta_i^{k+1}:=\Gamma^k_i (\beta_i^{k}, w^k, \theta_i), \forall k\in \{0,\cdots,K-1\}$, are a discrete-time Markov chain where the \textit{extrinsic} information at stage $k$ is characterized by the transition function $\Gamma^k_i(\cdot,w^k, \theta_i)$. 
Note that the transition function only characterizes how players update their beliefs at each stage yet does not guarantee that a player can learn the true types of others. 
The following example illustrates a class of players whose belief dynamics \textcolor{black}{exhibit} the confirmation bias \cite{nickerson1998confirmation} \textcolor{black}{where players tend to ignore intrinsic evidence such as $u^k$ and preserve their belief update rules $\Gamma_i^k$ at each stage $k$}. 
\begin{example}
\label{example:confirmation bias}
Consider a two-person game $N=2$ where the first player has two types $N_1=2, \Theta_1=\{\theta_1^1,\theta_1^2\}$ and the second player only has one type $N_1=1, \Theta_2=\{\theta_2^1\}$. 
The second player's belief state $\beta_2^k=[l_2^k(\theta_1^1|\theta_2^1),l_2^k(\theta_1^2|\theta_2^1)]$ toward the first player's type belongs to a finite set $\Lambda_2=\{[0.2,0.8],[0.5,0.5],[0.8,0.2]\}$. The transition function $\Gamma_2^k$ is independent of $k$: if the current belief state is $[0.5,0.5]$, then the belief at \textcolor{black}{the} next stage is $[0.2,0.8],[0.5,0.5],$ or $[0.8,0.2]$ with probability $0.4,0.2,0.4$, respectively. If the current belief state is $[0.8,0.2]$ (resp. $[0.2,0.8]$), then the belief at \textcolor{black}{the} next stage is $[0.8,0.2]$  (resp. $[0.2,0.8]$) or $[0.5,0.5]$  with probability $0.9$ and $0.1$,  respectively. 
The above transition function $\Gamma_2^k$ means that the second player tends to interpret the extrinsic information of the first player's type based on his current belief. 
If the second player already believes that the first player is of type $\theta_1^1$ with a high probability of $0.8$ at stage $k$,  i.e., \textcolor{black}{$\beta_2^k=[0.8,0.2]$}, then \textcolor{black}{the second player} is more inclined to enhance his current belief,  i.e., his belief \textcolor{black}{state} at \textcolor{black}{the} next stage, \textcolor{black}{i.e.,  $\beta_2^{k+1}$,} will remain \textcolor{black}{to be} $[0.8,0.2]$ with a high probability of $0.9$. 
The above transition function \textcolor{black}{represents} the \textcolor{black}{phenomena} of \textcolor{black}{attitude polarization and confirmation bias where players preserve their existing beliefs and the disagreement becomes more extreme at each stage even when players are exposed to the same evidence.}
\end{example}
}

 \subsection{Nonzero-Sum Cost Function and Equilibrium Concept}
 \label{subsec:cost}
At non-terminal stage $k\in \mathcal{K}\setminus \{K\}$, player $i$'s cost function is $g_i^k: \mathbb{R}^{n \times 1} \times \prod_{{j}=1}^N \mathbb{R}^{m_{{j}}\times 1} \times \Theta_i \mapsto \mathbb{R}$. The final stage cost is $g_i^K: \mathbb{R}^{n \times 1} \times \Theta_i \mapsto \mathbb{R}$. 
Define $u_i^{k_0:K-1}:=[u_i^{k_0}, \cdots, u_i^{K-1}]$ as player $i$'s action sequence from stage $k_0$ to $K-1$ and $u^{k_0:K-1}:=[u_i^{k_0:K-1},u_{-i}^{k_0:K-1}]$
as player $i$'s and all other players' action sequences from stage $k_0$ to $K-1$. 
Player $i$'s  expected cumulative cost from arbitrary initial stage $k_0\in \mathcal{K}$ to the terminal stage $K$ is defined as 
\begin{equation}
\label{eq:cost}
\begin{split}
&J_i^{k_0}(l_i^{k_0:K-1},u^{k_0:K-1}, x^{k_0}, \theta_i)=
\mathbb{E}_{w^{{K-1}} \sim \Xi_w}[g^{{K}}_i(x^K,\theta_i)]\\
& \quad\quad\quad\quad\quad\quad
+ \sum_{{k}={k_0}}^{K-1} \mathbb{E}_{w^{{k-1}} \sim \Xi_w}\left[\mathbb{E}_ {{\theta}_{-i}\sim {l}_i^{{k}}} [ g^{{k}}_i(x^k,u^k,\theta_i)] \right].
\end{split}
\end{equation}
The expectations are taken first over the external noise sequence $w^k$ and then over other players' internal type uncertainty. 
{
We cannot exchange the order of these two expectations as $l_i^k$ is a function of $w^{k-1}$.}  
Each player $i$ at stage $k_0\in \mathcal{K}$ aims to minimize $J_i^{k_0}$ by {choosing only his action sequence $u_i^{k_0:K-1}$ but not other players' action sequence $u_{-i}^{k_0:K-1}$.  
The following definition of sequential rationality in Definition \ref{def:sequentialRationality} guarantees that each player $i$ has no motivation to deviate from the sequentially rational action at any stage $k\in \{k_0,\cdots,K-1\}$ during the interaction if all other players adopt the sequentially rational actions.  
}
\begin{definition}[\textbf{Sequential Rationality}]
\label{def:sequentialRationality}
{An action sequence $u^{*,k_0:K-1}:=\{u_i^{*,k_0:K-1},u_{-i}^{*,k_0:K-1}\}$ is called \textit{sequentially rational} for player $i$ under the belief sequence $l_i^{k_0:K-1}$, state $x^{k_0}$, and type $\theta_i$, 
if for any state $x^k$ at stage $k\in \{k_0, \cdots, K-1\}$, 
player $i$ does not benefit from taking any other action sequence $u_i^{k:K-1}$, i.e., 
$J_i^{{k}}(l_i^{k:K-1}, u_i^{*,k:K-1},u_{-i}^{*,k:K-1}, x^{k} ,\theta_i)\leq J_i^{{k}}(l_i^{k:K-1}, u_i^{k:K-1},u_{-i}^{*,k:K-1}, x^{k} ,\theta_i), \forall u_i^{k:K-1}$.} 
\end{definition}
{Since players' actions may affect their future beliefs as captured by the belief dynamics $\Gamma_i^k$ in \eqref{eq:belief dynamics}, we further require the equilibrium action $u^{*,k_0:K-1}$ in Definition \ref{def:sequentialRationality} to be consistent with the belief dynamics, which leads to the following definition of Perfect Bayesian Nash Equilibrium (PBNE).}  
\begin{definition}[\textbf{Perfect Bayesian Nash Equilibrium}]
\label{def:PBNE}
Consider the $N$-player dynamic game of {private types and asymmetric information} defined by the state dynamics \eqref{eq:state dynamics} and the expected cumulative cost \eqref{eq:cost}. 
The action sequence $u^{*,0:K-1}:=\{u^{*,0:K-1}_i,u^{*,0:K-1}_{-i}\}$ of $N$ players over $K$ stages compromises the Perfect Bayesian Nash Equilibrium (PBNE) if, regardless of each player $i$'s type $ \theta_i\in \Theta_i$, \textcolor{black}{the following statements hold}. 
\begin{enumerate}
    \item \textbf{Sequential rationality}: $u^{*,0:K-1}$ is sequential rational for each player $i\in \mathcal{N}$ under his belief sequence  $l_i^{*,0:K-1}$;  
    \item \textbf{Belief consistency}: each player $i$'s belief sequence $l_i^{*,0:K-1}$ is consistent with \eqref{eq:belief dynamics} under $u^{*,0:K-1}$. 
\end{enumerate}
\end{definition}
{
\begin{proposition}
\label{prop: a feedback policy}
It is sufficient to represent player $i$'s equilibrium cost $J_i^{{k}}(l_i^{*,k:K-1}, u^{*,k:K-1}, x^{k} ,\theta_i)$ under the PBNE action $u^{*,k:K-1}$ at stage $k\in \mathcal{K}$ as a function of $\beta^k$, $x^k$ and $\theta_i$, \textcolor{black}{which is defined} as $V_i^{k}(\beta^{k}, x^{k} ,\theta_i)$. 
\textcolor{black}{Under} the boundary condition $V_i^{K}(\beta^K,x^{K} ,\theta_i):=g_i^K(x^K, \theta_i)$, 
the following holds for all $ k\in \{0,\cdots,K-1\}$ and all $x^k\in \mathbb{R}^{n\times 1}, \beta^k\in \Lambda$, i.e., 
\begin{equation}
\label{eq:DP}
\begin{split}
 & V_i^{k}(\beta^{k}, x^{k} ,\theta_i)  =\min_{u_i^k}  \sum_{{\theta}_{-i} }l_i^k({\theta}_{-i} |h^k, \theta_{i}) \{ g_i^k(x^k, u^k, \theta_i) + \\
 &\quad \quad \quad  \mathbb{E}_{w^k\sim \Xi_w}[V_i^{k+1}({\beta}^{k+1}, x^{k+1},\theta_i)]  \}, \forall \theta_i\in \Theta_i, \forall i\in\mathcal{N},  
\end{split}
\end{equation}
where $\beta^{k+1}$ and $x^{k+1}$ satisfy \eqref{eq:belief dynamics} and \eqref{eq:state dynamics}, respectively. 
\end{proposition}
\begin{proof}
According to the definition of PBNE, at the second last stage $k=K-1$, each player $i$'s equilibrium action $u_i^{*,k}=  
arg\min_{u_i^k} \mathbb{E}_{\theta_{-i}\sim l_i^k} [g_i^k(x^k, u^k, \theta_i)] +
\mathbb{E}_{w^k\sim \Xi_w}[g_i^{K}( x^{K},\theta_i)]$ is in general a function of $\theta_i,x^k,l_i^{*,k},u_{-i}^{*,k}$. Due to the coupling between $u_i^{*,k}$ and $u_{-i}^{*,k}$, we need to solve a  set of system equations for all $i\in \mathcal{N}$ and $\theta_i\in \Theta_i$. Then, $u_i^{*,k}$ will be a function of $\beta^k,x^k,\theta_i$ and we \textcolor{black}{obtain} \eqref{eq:DP} at stage $k=K-1$. We can repeat the above procedure from  $k=K-2$ to $k=0$ to obtain the recursive form in \eqref{eq:DP}.  
\end{proof}
Proposition \ref{prop: a feedback policy} characterizes the structure of the equilibrium action $u_i^{*,k}$ and the equilibrium cost $V_i^{k}(\beta^{k}, x^{k} ,\theta_i)$ for each player $i$ of type $\theta_i$ under the solution concept of PBNE; i.e., both terms are feedback functions of the belief state $\beta^k$, the physical state $x^k$, and the player' type $\theta_i$. 
Although $J_i^k$ is a function of beliefs  $l_i^{k:K-1}$ over all the remaining stages, $V_i^{k}(\beta^{k}, x^{k} ,\theta_i)$ only depends on the belief state at the current stage $k$. 
}
{
If all players' types are \textit{common knowledge}, 
PBNE still applies and we can define a new function $ \bar{V}_i^k(x^k,\theta)$ to represent the resulting equilibrium cost $ {V}_i^k(\beta^k,x^k,\theta_i)$ for all $k\in \mathcal{K}$ without loss of generality. 
}

\subsection{Offline Evaluation of Equilibrium Cost}
{
If each player $i$'s initial belief confirms to the prior distribution of other players' types, i.e., $l_i^0(\theta_j|x^0,\theta_i)=\Xi_{j}(\theta_j), \forall \theta_i\in \Theta_i, j\in \mathcal{N}, \theta_j\in \Theta_j, \forall x^0$, then each player $i$ at system state $x^0$ with belief state $\beta^0$
can use his expected equilibrium cost $\mathbb{E}_{\theta_i\sim \Xi_{i}}[{V}_i^0(\beta^0,x^0,\theta_i)]$ over his type uncertainty $\Xi_i$ as an offline performance measure of the equilibrium action $u^{*,0:K}$. 
As a comparison, player $i$'s expected equilibrium cost $\mathbb{E}_{\theta\sim \Xi}[\bar{V}_i^0(x^0,\theta)]$ under the complete information game serves as a benchmark. 
Note that player $i$ does not need to know the realization of the joint type $\theta$ to compute $\mathbb{E}_{\theta\sim \Xi}[\bar{V}_i^0(x^0,\theta)]$. 
Due to the coupling in dynamics, costs, and cognition among $N$ players, obtaining more information and knowing the type of another player $j\in\mathcal{N}\setminus\{i\}$
may not always improve player $i$'s performance; i.e.,  there is no guarantee that $\mathbb{E}_{\theta_i\sim \Xi_{i}}[{V}_i^0(\beta^0,x^0,\theta_i)] 
\geq \mathbb{E}_{\theta \sim \Xi}[\bar{V}_i^0(x^0,\theta)]$. 
Besides the above performance evaluation for an individual player $i\in \mathcal{N}$ under deception, we may also aim to evaluate the overall performance of multiple players or all $N$ players.}
We define the Price of Deception (PoD) in Definition \ref{def:PoD} with a set of coefficients $\eta_i \in [0,1], \forall i\in \mathcal{N}, \sum_{i\in \mathcal{N}}\eta_i=1$. 
{Since the equilibrium cost can be negative, we let $\eta_0(\Xi):=-\min ( 0, \allowbreak
\{\mathbb{E}_{\theta_i\sim \Xi_{i}}[{V}_i^0(\beta^0,x^0,\theta_i)]\}_{i\in \mathcal{N}}, \allowbreak 
\{\mathbb{E}_{\theta \sim \Xi}[\bar{V}_i^0(x^0,\theta)]\}_{i\in \mathcal{N}} )$ be the normalizing constant to guarantee that $p^{\eta}(\Xi)$ is non-negative for all chosen coefficients $\eta_i, i\in \mathcal{N}$.}  
\begin{definition}[\textbf{Price of Deception}]
\label{def:PoD}
For a given set of coefficients $\eta:=\{\eta_i\}_{i\in \mathcal{N} \cup \{0\}}$, the Price of Deception (PoD) of the $N$-player $K$-stage game defined by \eqref{eq:state dynamics}, \eqref{eq:cost}, and \eqref{eq:belief dynamics} under the prior probability distribution $\Xi=[\Xi_i]_{i\in \mathcal{N}}$ is 
\begin{equation*}
p^{\eta}(\Xi) := \frac{\sum_{i\in \mathcal{N}} \eta_i 
\mathbb{E}_{\theta\sim \Xi}[\bar{V}_i^0(x^0,\theta)]
+\eta_0(\Xi) }{\sum_{i\in \mathcal{N}} \eta_i 
\mathbb{E}_{\theta_i\sim \Xi_{i}}[ {V}_i^0(\beta^0,x^0,\theta_i)]
+\eta_0(\Xi) }\in [0,\infty). 
\end{equation*} 
\end{definition}
The PoD is a crucial evaluation and design metric. 
{We can endow PoD with different meanings by properly choosing the weighting coefficients $\eta_i, i\in \mathcal{N}$. 
For example, if besides $N$ players, there is a central planner who 
aims to minimize the total cost of all $N$ players under their deceptive interaction. Then, we can pick $\eta_i=1/N, i\in \mathcal{N}$, to represent the overall system performance. 
Although the central planner cannot control players' state dynamics, costs, and belief dynamics directly, he can still affect their deceptive interaction if he can design the prior probability distribution $\Xi$ of the joint type $\theta$. 
If the  central planner instead only aims to reduce the cost of one player $j\in \mathcal{N}$, then we can pick $\eta_j=1$ and $\eta_h=0, \forall h\in \mathcal{N}\setminus \{j\}$.}  
With a given weighting parameters $\eta$, a larger value of $p_{\eta}(\Xi)$ indicates a better accomplishment of the above goals. 
Note that individual deception may improve the system performance, i.e., $p^{\eta}(\Xi)>1$. 

\section{Linear-Quadratic Specification }
\label{sec:LQ}
{
Linear-Quadratic (LQ) game is an important class of dynamic games.  
They can also be applied iteratively to approximate nonlinear stochastic systems with general cost functions and obtain equilibrium actions \cite{fridovich2019iterative}.}  
In the following sections, we consider linear state dynamics 
\begin{equation}
\label{eq:linearDynamics}
f^k(x^k, u^k,\theta) :=A^k(\theta)x^k+ \sum_{{i}=1}^N B^k_{{i}}(\theta_{{i}}) u_{{i}}^k, 
\end{equation}
 with stage-varying matrices $A^k(\theta) \in \mathbb{R}^{n\times n }$, $B^k_i(\theta_i)\in \mathbb{R}^{n\times m_i}$. 
{\begin{remark}
\label{thm:LQcontrollable}
System \eqref{eq:linearDynamics} is multi-agent controllable if and only if matrices $H_i^k(\theta):=[B_i^{k-1}(\theta_i), \allowbreak
\cdots, \prod_{{h}=2}^{k-1} A^{{h}} (\theta)B_i^1(\theta_i), \allowbreak
\prod_{{h}=1}^{k-1} A^{{h}}(\theta) B_i^0(\theta_i)], \forall i\in \mathcal{N}, \forall \theta\in \Theta, \forall k\in \mathcal{K}$, are of full rank as noise $w^k$ has zero mean and we can obtain $\mathbb{E} [x^{k}]=\prod_{{h}=0}^{k-1} A^{{h}}(\theta) x^0 \allowbreak +
\sum_{{r}=1}^N H_{{r}}^k(\theta) [u_{{r}}^{k-1};\cdots; u_{{r}}^{0}]$ by induction. 
\end{remark}}
Each player $i$'s cost is quadratic in both $x^k$ and $u^k$; i.e., 
\begin{equation}
\begin{split}
\label{eq:LQcost}
g_i^k(x^k, u^k,\theta_i)=(x^k-\hat{x}_{i}^k(\theta_i))' D_i^k(\theta_i) (x^{k}-\hat{x}_{i}^k(\theta_i) )\\
+\hat{f}_{i}^k(\hat{x}_{i}^k(\theta_i)) + 
\sum_{j=1}^N (u^{k}_j)' F^k_{ij} (\theta_i)u^{k}_j, \forall k\in \mathcal{K},    
\end{split}
\end{equation}
where $[\hat{x}_{i}^k(\theta_i) ]_{k\in \mathcal{K}}$ is a known type-dependent reference trajectory for player $i\in \mathcal{N}$ and $\hat{f}_{i}^k$ is a known function of $\hat{x}_{i}^k(\theta_i)$. 
The cost matrices $D_i^k(\theta_i)\in \mathbb{R}^{n \times n }, F_{ij}^k(\theta_i) \in \mathbb{R}^{m_i \times m_i } , \forall i,j\in \mathcal{N},k\in \mathcal{K}$,  are symmetric. 
At the final stage,  $F^K_{ij}(\theta_i) \equiv \mathbf{0}_{m_i , m_i }, \forall i,j\in \mathcal{N}, \forall \theta_i\in \Theta_i$. 
{We introduce the following three sets of notations for the belief matrix, the extended Riccati \textcolor{black}{equations}, and the matrix-form equilibrium action, respectively.} 

\paragraph{Belief Matrix}
 With a little abuse of notation, 
we can define the marginal probability $l_i^k({\theta}_{j} | h^k, \theta_i):=\sum_{\theta_r\in \Theta_r,r\in \mathcal{N}\setminus \{i,j\} } l_i^k(\theta_{-i} | h^k, \theta_i), \allowbreak
\forall j\in \mathcal{N}\setminus \{i\}$, as the player $i$'s belief toward the player $j$'s type at stage $k$. 
Define the belief matrix for all $i\in \mathcal{N}, j\in \mathcal{N}\setminus \{i\},k\in \{0,\cdots,K-1\}$, \textcolor{black}{as} 
 \begin{equation}
 \label{eq:belief matrix}
 \mathbf{L}_{ij}^{k}:=
\begin{bmatrix}
&\mathbf{L}_i^k(\theta_{j}^1|h^k, \theta_i^1),   &\cdots 
& \mathbf{L}_i^k(\theta_{j}^{N_j}|h^k, \theta_i^1)  \\
& \mathbf{L}_i^k(\theta_{j}^1|h^k, \theta_i^2)  , &\cdots
& \mathbf{L}_i^k(\theta_{j}^{N_j}|h^k, \theta_i^2)  \\
&\vdots &\ddots  &\vdots  \\
& \mathbf{L}_i^k(\theta_{j}^1|h^k, \theta_i^{N_i})   , &\cdots
& \mathbf{L}_i^k(\theta_{j}^{N_j}|h^k, \theta_i^{N_i}) 
\end{bmatrix}, 
 \end{equation}
where each block element 
{$
\mathbf{L}_i^k(\theta_{j}^r|h^k, \theta_i^h)=\Di[{l}_i^k(\theta_{j}^r|h^k, \theta_i^h), \allowbreak
\cdots, \allowbreak
{l}_i^k(\theta_{j}^r|h^k, \theta_i^h)]\in \mathbb{R}^{n \times n }, \forall r\in\{1,\cdots,N_j\}, \forall h\in\{1,\cdots,N_i\}. 
$ }
Since all its elements are positive and all rows sum to one, the belief  matrix $\mathbf{L}_{ij}^{k}$ is a \textit{right stochastic matrix}. 

\paragraph{Extended Riccati Equations}
Let a sequence of symmetric matrices $S_i^k(\beta^{k},\theta_i) \in \mathbb{R}^{n\times n}$, 
vectors $N_i^k(\beta^{k},\theta_i) \in \mathbb{R}^{n\times 1}$, and scalars $q_i^k(\beta^{k},\theta_i) \in \mathbb{R}$ satisfy the following extended Riccati equations 
for all $ \beta^{k}\in \Lambda, i\in \mathcal{N},  \theta_i\in \Theta_i, k\in \{0,\cdots,K-1\}$: 

\begin{equation}
    \label{eq:S}
    \resizebox{.89\hsize}{!}{$
\begin{split}
 S_i^k=&D_i^k+ \mathbb{E}_{\theta_{-i}\sim l_i^k}\bigg[ (A^k+\sum_{j=1}^N B^k_j\Psi^{1,k}_j )' \mathbb{E}_{w^k\sim \Xi_w}[S_i^{k+1}]  \\
 & \cdot (A^k+ \sum_{j=1}^N B^k_j \Psi^{1,k}_j) +\sum_{j=1}^N (\Psi^{1,k}_j )'F_{ij}^k \Psi^{1,k}_j  \bigg],
\end{split}
$}
\end{equation}

\begin{equation}
    \label{eq:N}
    \resizebox{.89\hsize}{!}{$
\begin{split}
 N_i^k=&-2D^k_i\hat{x}_{i}^k+ \mathbb{E}_{\theta_{-i}\sim l_i^k}\bigg[  (\sum_{j=1}^N B^k_j \Psi^{1,k}_j  +  A^k)' ( \mathbb{E}_{w^k\sim \Xi_w}[ N_i^{k+1}] \\
 & + 2 \mathbb{E}_{w^k\sim \Xi_w}[S_i^{k+1}]\sum_{j=1}^N B^k_j \Psi^{2,k}_j  )  + 2\sum_{j=1}^N (\Psi^{1,k}_j  )'F_{ij}^k\Psi^{2,k}_j  \bigg] , 
\end{split}
$}
\end{equation}

\begin{equation}
    \label{eq:q}
    \resizebox{.89\hsize}{!}{$
\begin{split}
 q_i^k= &(\hat{x}_{i}^k)'D^k_i \hat{x}_{i}^k+\hat{f}_{i}^k(\hat{x}_{i}^k) + \mathbb{E}_{w^k\sim \Xi_w}[(w^k)'S_i^{k+1}w^k+q_i^{k+1}] \\
&+\mathbb{E}_{\theta_{-i}\sim l_i^k}\bigg[ (\sum_{j=1}^N B^k_j \Psi^{2,k}_j)' \mathbb{E}_{w^k\sim \Xi_w}[S_i^{k+1}]\sum_{j=1}^N B^k_j \Psi^{2,k}_j\\
& + (\sum_{j=1}^N B^k_j \Psi^{2,k}_j)' \mathbb{E}_{w^k\sim \Xi_w}[N_i^{k+1}]   +\sum_{j=1}^N ( \Psi^{2,k}_j )'F_{ij}^k\Psi^{2,k}_j \bigg],    
\end{split}
$}
\end{equation}
where functions $\Psi^{1,k}_i,\Psi^{2,k}_i, \forall i\in \mathcal{N}$, are defined below. 
The boundary conditions of the extended Riccati equations are 
\begin{equation}
\label{eq:boundary}
    S_i^K=D^K_i;\ N_i^K =-2D_i^K \hat{x}_{i}^K;\  q_i^K =(\hat{x}_{i}^K)' D_i^K \hat{x}_{i}^K+\hat{f}_{i}^K(\hat{x}_{i}^K). 
\end{equation}


\paragraph{Equilibrium Action in Matrix Form}
\label{para: action}
We need to represent the equilibrium action of all players under all types in matrix form as each player's action is coupled with other players' actions under PBNE. 
Since each player $i$ has different equilibrium actions under different types, with a little abuse of notation, we write each player $i$'s action as a function of his type $\theta_i$ and define two action vectors
$ \mathbf{u}_i^k:=[u_i^{k}(\theta^1_i),\cdots,  u_i^{k}(\theta^{N_i}_i)]' \in \mathbb{R}^{m_i N_i\times 1}$ and $\mathbf{u}^k:=[ \mathbf{u}_1^k,\mathbf{u}_2^k\cdots,   \mathbf{u}_N^k]' \in \mathbb{R}^{\sum_{{{r}}=1}^N m_{{r}} N_{{r}} \times 1}$. 
For all $ i\in \mathcal{N}, l_i^k,  \theta_i\in \Theta_i, k\in \{0,\cdots,K-1\}$, define a series of $(m_i)$-by-$(m_i)$ square matrices 
\begin{equation*}
    R^{k}_i({\beta}^{k},\theta_i):=F_{ii}^k(\theta_i)+(B^k_i(\theta_i))' S_i^{k+1}(\beta^{k},\theta_i)B^k_i(\theta_i). 
\end{equation*}
Let $\mathbf{B}_i^k:=\Di[B_i^{k}(\theta_i^1) \cdots, B_i^{k}(\theta_i^{N_i})] $ be $(N_in)$-by-$(N_i m_i)$ block matrices and 
$\mathbf{S}_i^k(\beta^k):=\Di[S_i^{k}(\beta^{k},\theta_i^{1}), \cdots, S_i^{k}(\beta^{k},\theta_i^{N_i})]$ be $(N_in)$-by-$(N_in)$ block matrices. 
Finally,  define  parameter matrices $\mathbf{W}^{1,k}(\beta^k)=[W^{1,k}_1(\beta^k); \cdots;  W^{1,k}_N(\beta^k)] \in \mathbb{R}^{\sum_{{{r}}=1}^N m_{{r}} N_{{r}} \times n}$, $\mathbf{W}^{2,k}(\beta^k)=[W^{2,k}_1(\beta^k);\cdots; W^{2,k}_N(\beta^k)] \in \mathbb{R}^{\sum_{{r}=1}^N m_{{r}} N_{{r}} \times 1}$, and $\mathbf{W}^{0,k}(\beta^k):=[W^{0,k}_{ij}(\beta^k)\in \mathbb{R}^{ m_i N_i \times m_j N_j } ]_{i,j\in \mathcal{N}}$ for any $\beta^k\in \Lambda$. Their elements are given as follows; i.e., $\forall i\in \mathcal{N}, \forall k\in \{0,\cdots,K-1\}$, 
\begin{equation*}
\begin{split}
&W^{1,k}_i(\beta^k)  =\bigg[ 
(B^k_i(\theta^1_i))' S_i^{k+1}(\beta^{k},\theta^1_i)\mathbb{E}_{{\theta}_{-i}\sim {l}_i^k } [A^k(\theta^1_i, \theta_{-i})];\\
&
\quad \quad \quad \quad  \cdots;  (B^k_i(\theta^{N_i}_i))' S_i^{k+1}(\beta^{k},\theta^{N_i}_i)\mathbb{E}_{{\theta}_{-i}\sim {l}_i^k } [A^k(\theta^{N_i}_i, \theta_{-i})] 
\bigg], 
\\
& W^{2,k}_i(\beta^k)  = \frac{1}{2} \bigg[ (B^k_i(\theta^1_i))'N_i^{k+1}(\beta^{k}, \theta^1_i); \\
& \quad \quad \quad \quad \cdots;   (B^k_i(\theta^{N_i}_i))'N_i^{k+1}(\beta^{k}, \theta^{N_i}_i)  
\bigg], \\
&W_{ii}^{0,k}(\beta^k)=\Di[R^{k}_i (\beta^{k}, \theta_i^1), \cdots, R^{k}_i(\beta^{k}, \theta_i^{N_i})  ],  \\
& W_{ij}^{0,k}(\beta^k)=(\mathbf{B}_i^k)'\mathbf{S}_i^{k+1}(\beta^k) \mathbf{L}_{ij}^{k}\mathbf{B}_j^k, \forall j\in \mathcal{N}\setminus \{i\}.
\end{split}
\end{equation*}
Let matrix $\mathbf{M}_{i}^{k}(\beta^k,\theta_i^l)\in \mathbb{R}^{m_i\times \sum_{{r}=1}^N m_{{r}}N_{{r}}}, l\in \{1,2,\cdots,N_i\}, i\in \mathcal{N}, k\in \{0,\cdots,K-1\}$, be the truncated row block, i.e., from row $\sum_{{r}=1}^{i-1} m_{{r}}N_{{r}}+m_i(l-1)$ to $\sum_{{r}=1}^{i-1} m_{{r}}N_{{r}}+m_i l$, of matrix $(-\mathbf{W}^{0,k}(\beta^k) )^{-1}$. 
{
Define shorthand notations 
$\Psi^{1,k}_i(\beta^k, \theta_i):= \mathbf{M}_{i}^k(\beta^k,\theta_i)\mathbf{W}^{1,k}(\beta^k)$ and
$\Psi^{2,k}_i(\beta^k, \theta_i):=\mathbf{M}_{i}^k(\beta^k,\theta_i)\mathbf{W}^{2,k}(\beta^k)$.}

\subsection{Extrinsic Belief Dynamics and Extended Riccati Equations}
{In this section, we focus on the extrinsic belief dynamics where $\Gamma_i^k$ is independent of players' actions $u^k$ for all $i\in \mathcal{N},k\in  \{0,\cdots,K-1\}$.} 
The proof of Theorem \ref{thm:ValueFunction} generalizes the one of classical LQ \textcolor{black}{games} (e.g., Chapter $5.5$ and $6.2$ in  \cite{basar1999dynamic}) where we further incorporate players' {asymmetric belief dynamics} into their objective functions to \textcolor{black}{minimize their expected costs} under deception. 
We apply \textit{dynamic programming} from stage $K-1$ backward to stage $0$ to obtain a closed-form solution of PBNE. 
\begin{theorem}
\label{thm:ValueFunction}
An $N$-player $K$-stage LQ game of incomplete information defined by \eqref{eq:linearDynamics}, \eqref{eq:LQcost}, and extrinsic belief dynamics $\beta^{k+1}_i=\Gamma_i^k(\beta^{k}_i,w^k,\theta_i), \forall i\in \mathcal{N}, \forall k\in\{0,\cdots,K-1\}$, admits a unique state-feedback  PBNE  
\begin{equation}
\label{eq:control}
{
u_i^{*,k}(\beta^k,x^k, \theta_i)=\Psi^{1,k}_i(\beta^k, \theta_i) x^k+\Psi^{2,k}_i(\beta^k, \theta_i),} 
\end{equation}
 if and only if  $R_i^k(\beta^{k}, \theta_i)$ is positive definite and $\mathbf{W}^{0,k}(\beta^k)$ is non-singular for all $\beta^{k}\in \Lambda, i\in \mathcal{N}, \theta_i\in \Theta_i,k\in \{0,\cdots, K-1\}$. 
The equilibrium cost $V_i^k$ is quadratic in $x^k$, i.e., 
\begin{equation}
\label{eq:QudraticValuefunction}
\begin{split}
V_i^k({\beta}^{k}, x^k,\theta_i)&=q_i^k({\beta}^{k},\theta_i)+(x^k)' N_i^k({\beta}^{k},\theta_i)\\
&
+(x^k)' S^k_i({\beta}^{k},\theta_i)x^k, \forall i\in \mathcal{N}, k\in\mathcal{K}. 
\end{split}
\end{equation}
\end{theorem}

\begin{proof}
We use backward induction to prove the result. 
At the final stage $K$, the value function $V_i^K(\beta^K,x^K,\theta_i)=(x^K-\hat{x}_{i}^K(\theta_i))' D_i^K(\theta_i) (x^{K}-\hat{x}_{i}^K(\theta_i) )+\hat{f}_{i}^K(\hat{x}_{i}^K(\theta_i))$ is quadratic in $x^K$ and we obtain the boundary conditions for  $S_i^K,N_i^K,q_i^K$ in \eqref{eq:boundary} 
by matching the RHS of \eqref{eq:QudraticValuefunction}.  
At  any stage $ k\in \{0,\cdots,K-1\}$, if \eqref{eq:QudraticValuefunction} is true at stage $k+1$, 
{we can expand $\mathbb{E}_{w^k\sim \Xi_w}[V_i^{k+1}(\beta^{k+1}, x^{k+1},\theta_i)]$ by plugging in the state dynamics $x^{k+1}=A^k(\theta)x^k+ \sum_{i=1}^N B^k_i(\theta_i) u_i^k+w^k$ and the belief dynamics $\beta^{k+1}_i=\Gamma_i^k(\beta^{k}_i,w^k,\theta_i)$.} 
Then, the Right-Hand Side (RHS)  of \eqref{eq:DP} 
is quadratic in $u_i^k$ for each player $i$. 
{If the coefficient matrix $R_i^k$ of the quadratic form $(u_i^k)' R_i^k u_i^k$} is positive definite, then the first-order necessary conditions for minimization are also sufficient and we obtain the following unique set of equations for the equilibrium action  $u^{*,k}$ by differentiating the RHS of \eqref{eq:DP} and setting it to zero, i.e., $\forall \theta_i\in \Theta_i$, 
\begin{equation}
\label{eq: PBNE}
\begin{split}
& -R^{k}_i u_i^{*,k}(\theta_i)= (B_i^k)' S^{k+1}_i \mathbb{E}_{\theta_{-i}\sim l_i^k} [ A^k ]   x^k
+\frac{1}{2}(B_i^k)' N_i^{k+1}
\\
&\quad \quad +(B_i^k)' S_i^{k+1}\sum_{j\neq i}\mathbb{E}_{\theta_{j}\sim l_i^k} [ B^k_j(\theta_j) u^{*,k}_j(\theta_j) ], \forall i\in \mathcal{N}. 
\end{split}
\end{equation}
Due to the coupling in players' actions and beliefs, we rewrite \eqref{eq: PBNE} in matrix form, i.e.,  $ -\mathbf{W}^{0,k}(\beta^k)  \mathbf{u}^{*,k}=\mathbf{W}^{1,k}(\beta^k)  x^k +\mathbf{W}^{2,k}(\beta^k)$, to solve the set of equations.  
Given the existence of $(-\mathbf{W}^{0,k}(\beta^k))^{-1}$, each player $i$'s equilibrium action is an affine function in $x^k$, {i.e., $u_i^{*,k}(\beta^k,x^k, \theta_i)=\Psi^{1,k}_i(\beta^k, \theta_i) x^k+\Psi^{2,k}_i(\beta^k, \theta_i)$. Note that the coefficients $\Psi^{1,k}_i,\Psi^{2,k}_i$ for player $i$ are functions of $\beta^k$, i.e., the beliefs of all players under all types at stage $k$.} 
{
Finally, after substituting the equilibrium action $ u_i^{*,k}(\beta^k,x^k, \theta_i)=\Psi^{1,k}_i(\beta^k, \theta_i) x^k+\Psi^{2,k}_i(\beta^k, \theta_i)$ into the RHS of \eqref{eq:DP} and representing $V_i^k$ in the Left-Hand Side (LHS) in its quadratic form of $x^k$,}  we can match the coefficients of quadratic, linear, and constant terms in the LHS and RHS  to obtain the extended Riccati equations \eqref{eq:S}, \eqref{eq:N}, and \eqref{eq:q}.  
\end{proof}

\begin{remark}[\textbf{Positive Definiteness}]
{
If $D_i^k(\theta_i)$ and $F_{ij}^k(\theta_i), \forall j\in \mathcal{N}$, are positive definite for all $k\in \mathcal{K}$, then $R_i^k(\beta^k,\theta_i)$ is positive definite for all $k\in \mathcal{K},\beta^k\in \Lambda$, \textcolor{black}{because} the linear combination of positive definite matrices in \eqref{eq:S} preserves positive definiteness. 
Note that the above condition is only a necessary condition; i.e., $D_i^k$ and $F_{ij}^k$ do not need to be positive definite to make $R_i^k$ positive definite as shown in Section \ref{sec:case study}.} 
\end{remark}

\begin{remark}[\textbf{\textcolor{black}{Cognitive Coupling}}]
\label{remark:cognition coupling}
{
Compared with the classical LQ \textcolor{black}{games} (e.g., Chapter $6$ in \cite{basar1999dynamic}), the deception of players' types results in a unique feature of \textcolor{black}{\textit{cognitive coupling}} represented by the belief matrix in \eqref{eq:belief matrix}; i.e., each player's action hinges on not only his own belief but also all other players' beliefs as these beliefs can affect their actions and further the outcome of the interaction.} 
Thus, player $i$ can change other players' actions \textcolor{black}{by manipulating} their beliefs of his type $\theta_i$, i.e.,  $l_{j}^k, \forall j\in \mathcal{N}\setminus \{i\},$  or \textcolor{black}{making} them believe that his belief $l_i^k$ on their types $\theta_{-i}$ \textcolor{black}{has} changed. 
\end{remark}
{We introduce matrix block partitions as follows.}  
For each type $\theta_i\in \Theta_i$, we divide $A^k(\theta),D_i^k(\theta_i),S_i^k(\theta_i)$ into $N$-by-$N$ blocks where the $(i,i)$ block is $A_i^k(\theta),\bar{D}_i^k(\theta_i),\bar{S}_i^k(\theta_i)\in \mathbb{R}^{ n_i\times n_i} $, respectively. The $i$-${th}$ row block of $N_i^k(\theta_i),\hat{x}_{i}^k(\theta_i)$ is $\bar{N}_i^k(\theta_i),\bar{x}_{i}^k(\theta_i)\in \mathbb{R}^{n_i\times 1}$, respectively. 
The  $i$-${th}$ row block of $B_i^k(\theta_i)$ is $\bar{B}_i^k(\theta_i)\in \mathbb{R}^{n_i\times m_i }$. 
{
When the system state $x^k$ can be represented by players' joint states $[x_i^k]_{i\in \mathcal{N}}$, Corollary \ref{corollary: degeneration} shows that the LQ game of asymmetric information degenerates to an LQ control problem if
players have decoupled cost and state dynamics defined as follows.} 
\begin{definition}[\textbf{Decoupled Dynamics and Cost}]
\label{def:decouple}
 Player $i\in \mathcal{N}$ has decoupled dynamics if for all $k\in \mathcal{K}$, $A_i^k(\theta)=\bar{A}_i^k(\theta_i), \forall \theta\in \Theta$, while all other elements in the $i$-${th}$ row block and the $i$-${th}$ column block of $A^k(\theta)$ are $0$. Besides, all elements of $B_i^k(\theta_i)$ except for the row block $\bar{B}_i^k(\theta_i)$ are required to be $0$.   
Player $i\in \mathcal{N}$ has a decoupled cost if for all stage $k\in \mathcal{K}$, $F^k_{ij}(\theta_i)=\mathbf{0}_{m_i,m_i}, \forall \theta_i\in \Theta_i, j\in \mathcal{N}\setminus \{i\}$, and all elements of $D_i^k(\theta_i)$  equal $0$ except for $\bar{D}_i^k(\theta_i)$. 
\end{definition}

\begin{corollary}[\textbf{Degeneration to LQ Control}]
\label{corollary: degeneration}
{
If $x^k=[x_i^k]_{i\in \mathcal{N}}$
for all stage $k\in \mathcal{K}$ and player $i$ has both decoupled cost and  state dynamics, then his action under PBNE is independent of other players' actions, types, and beliefs, i.e.,} 
$u_i^{*,k}=-(R_i^k)^{-1}(\bar{B}_i^k)'\bar{S}_i^{k+1}A_i^k x_i^k-\frac{1}{2} (R_i^k)^{-1}(\bar{B}_i^k)'\bar{N}_i^{k+1}$, where $R_i^k=F_{ii}^k+(\bar{B}_i^k)'\bar{S}_i^{k+1}\bar{B}_i^k$, $(G_i^k)'=\mathbf{I}_n-\bar{S}_i^{k+1}\bar{B}_i^k(R_i^k)^{-1}(\bar{B}_i^k)'$, $\bar{S}_i^k=(A_i^k)'(G_i^k)'\bar{S}_i^{k+1} A_i^k+\bar{D}_i^k$, and $\bar{N}_i^k=(A_i^k)'(G_i^k)'\bar{N}_i^{k+1}-2\bar{D}_i^k\bar{x}^k_{i}$. 
\end{corollary}

\begin{proof}
We show by induction that $S_i^k,N_i^k, \forall k\in \mathcal{K},$ satisfy the \textit{sparsity condition} that only the $(i,i)$ block of $S_i^k$ and the $i$-${th}$ row block of $N_i^k$ are nonzero. 
At stage $K$, $S_i^K=D_i^K$ and $N_i^K=-2D_i^K\hat{x}_{i}^K$ satisfy the above  condition. 
At stage $k\in \{0,\cdots,K-1\}$, if $S_i^{k+1},N_i^{k+1}$ satisfy the sparsity condition, $\mathbf{W}^{0,k}(\beta^k)$ becomes a diagonal block matrix where $W^{0,k}_{ij}(\beta^k)=\mathbf{0}_{m_iN_i,m_jN_j}$ and $\mathbf{M}_{i}^k(\beta^k,\theta_i)=-(R_i^k(\beta^k,\theta_i))^{-1}$ for all $\beta^k\in \Lambda$. Then, $S_i^k,N_i^k$ satisfy the condition  based on \eqref{eq:S} and \eqref{eq:N}. 
\end{proof}

\subsection{Intrinsic Belief Dynamics and Receding-Horizon Control}
{
If there exists a player $i\in\mathcal{N}$ whose belief dynamics $\Gamma_i^k$ depend on intrinsic information at some stage $k\in \{0,\cdots,K-1\}$ as shown in \eqref{eq:belief dynamics}, then the equilibrium action $u_i^{*,k}$ is in general a nonlinear function of $x^k$ and the equilibrium cost $V_i^k$ is not quadratic in $x^k$ even under the LQ setting of \eqref{eq:linearDynamics} and \eqref{eq:LQcost}. 
Besides the \textcolor{black}{\textit{static cognitive coupling}} among $N$ players in Remark \ref{remark:cognition coupling}, the intrinsic information of $u^k$ in the belief update introduces another \textcolor{black}{\textit{dynamic cognitive coupling}} between the forward belief dynamics via \eqref{eq:belief dynamics} and the backward equilibrium computation via \eqref{eq:DP}, which makes it challenging to compute PBNE.}  
To reduce the computational complexity and further obtain \textcolor{black}{implementable} actions,  we adopt a receding-horizon approach \textcolor{black}{that} computes the sequentially rational action sequence of all the future stages $u^{*,k:K-1}$ at current stage $k\in \{0,\cdots,K-1\}$ assuming  $\beta^{\bar{k}}=\beta^{k}, \forall \bar{k}\in \{k,...,K-1\}$, yet only \textcolor{black}{implements} the current-stage action $u^{*,k}$. 
Then, at the new stage $k+1$, each player observes the new system state $x^{k+1}$ and updates the belief to $\beta^{k+1}$ and recomputes the entire action sequence $u^{*,k+1:K-1}$ under assumption of  $\beta^{\bar{k}}=\beta^{k+1}, \forall \bar{k}\in \{k+1,...,K-1\}$, yet still only \textcolor{black}{implements} the new current-stage action $u^{*,k+1}$. Players repeat the above procedure until they  reach the final stage of the interaction. 

Compared with PBNE, which produces an offline planning for all future stages under all possible scenarios before the game has taken place, 
the receding-horizon approach enables an online replanning of their actions repeatedly at the beginning of each new stage as the interaction continues. 
{
Although we assume that players' beliefs at the future stages are the same as the current beliefs during the phase of equilibrium computation, players can correct and update their beliefs and actions based on the online observation of $x^k$ during each replanning phase. 
Thus, the receding-horizon approach provides a reasonable approximation of the PBNE action and is more adaptive to unexpected environmental changes of the state dynamics $f^k$ and cost structure $g^k_i, \forall i\in\mathcal{N}$.}   

Under the LQ specification in \eqref{eq:linearDynamics} and \eqref{eq:LQcost} and Bayesian belief dynamics in \eqref{eq:Bayesian}, we summarize the computation phase and online implementation phase in Algorithm \ref{algorithm:computation} and \ref{algorithm:implentation}, respectively. 
To investigate the scalability of our algorithms, 
we analyze the temporal and spatial complexity concerning $N,K$, and $N_i$. To simplify the notation and enhance readability, we focus on the symmetric setting where $N_i=N_0\in \mathbb{Z}^+, \forall i\in \mathcal{N}$. 
For each player $i\in\mathcal{N}$ of type $\theta_i\in \Theta_i$ at the beginning of the interaction, i.e., $k=0$, he needs to store the game parameters $A^0,B_{{r}}^0({\theta}_{{r}}),D_{{r}}^0({\theta}_{{r}}),F_{{r}h}^0({\theta}_{{r}}), \forall  {\theta}_{{r}}\in \Theta_{{r}}$,
and the belief matrix $\mathbf{L}_{{r}h}^0$ for all $ {r},h\in \mathcal{N}$, which are \textit{common knowledge}. 
The spatial complexity to store the game parameters and the belief matrix is  $O(N^2 N_0)$ and  $O(N^2 N_0^2)$, respectively. 
Note that in general, player $i$ has coupled cognition as shown in Remark \ref{remark:cognition coupling} and has to keep track of not only his belief $\mathbf{L}_{i,j}^k,\forall j\in \mathcal{N}$, but also other players' beliefs $\mathbf{L}_{{r},h}^k,\forall {r}\in \mathcal{N}\setminus \{i\},h\in \mathcal{N}$,  to decide his equilibrium action under deception at each stage $k$. 
During the $K$-stage interaction, each player $i\in\mathcal{N}$ of type $\theta_i\in \Theta_i$ observes the system state $x^k$ and computes his equilibrium action $u_i^{*,k}(\beta^k,x^k,\theta_i)$ at stage $k$ based on Algorithm \ref{algorithm:computation}. 
After all players implement their equilibrium actions at stage $k$, the system state evolves to $x^{k+1}$. Based on the new state observation $x^{k+1}$, each player $i$ updates the belief matrix in \eqref{eq:belief matrix} via \eqref{eq:Bayesian}. 
Since player $i$ can delete the game parameters and the belief matrices of previous stages, the spatial complexity remains the same as the real-time stage index $k$ increases.  
Thus, our algorithm can handle the interaction of long duration. 
All players repeat the above procedure stated in lines $14$-$17$ of Algorithm \ref{algorithm:implentation}  until reaching the terminal stage $k=K$.

\begin{algorithm}[]
\SetAlgoLined
 \textbf{Load} game parameters $A^k,B_{{r}}^k(\bar{\theta}_{{r}}),D_{{r}}^k(\bar{\theta}_{{r}}),F_{{r}h}^k(\bar{\theta}_{{r}}), 
 \allowbreak
 \forall  \bar{\theta}_{{r}}\in \Theta_{{r}}$
 and the belief matrix $\mathbf{L}_{{r},h}^k$ 
  for all $ {r},h\in \mathcal{N}$\;
 \textbf{Input} state observation $x^k$\;
   \For{$\bar{k} \leftarrow K-1$ \KwTo $k$ }
   {
  \For{$j \leftarrow 1$ \KwTo $N$ }
   {
     \For{$\theta_j \leftarrow \theta_j^1$ \KwTo $\theta_j^{N_j}$  }
     {
   Compute $S_j^{\bar{k}},N_j^{\bar{k}}$ via \eqref{eq:S},  \eqref{eq:N} with $\beta^{\bar{k}}=\beta^{k}$\;
   }
   }   
  }
  \textbf{Return} his equilibrium action  $u_i^{*,k}(l_i^k,x^k,\theta_i)$ via \eqref{eq:control}\; 
 \caption{PBNE computation with $\beta^{\bar{k}}=\beta^{k}, \forall \bar{k}\in \{k,...,K-1\}$ at stage $k \in \{0,\cdots,K-1\}$ for player $i\in \mathcal{N}$ of type $\theta_i\in \Theta_i$  
 \label{algorithm:computation}}
\end{algorithm}

\begin{algorithm}[]
\SetAlgoLined
 \textbf{Initialize} $k=0$\;
 \textbf{Store} game parameters $A^k,B_{{r}}^k(\bar{\theta}_{{r}}),D_{{r}}^k(\bar{\theta}_{{r}}),F_{{r}h}^k(\bar{\theta}_{{r}}), 
 \allowbreak
 \forall  \bar{\theta}_{{r}}\in \Theta_{{r}}$
 and the belief matrix $\mathbf{L}_{{r},h}^k$ 
  for all $ {r},h\in \mathcal{N}$\;
  \While{$k<K$}
  { 
  Call Algorithm \ref{algorithm:computation} to  {implement}  $u_i^{*,k}(l_i^k,x^k,\theta_i)$\; 
  {Observe} state $x^{k+1}$ and {update} all elements of the belief matrix via \eqref{eq:Bayesian} to obtain $\mathbf{L}_{{r},h}^{k+1}, \forall {r},h\in \mathcal{N}$\; 
  \textbf{Delete} $A^k,B_{{r}}^k(\bar{\theta}_{{r}}),D_{{r}}^k(\bar{\theta}_{{r}}),F_{{r}h}^k(\bar{\theta}_{{r}}),\mathbf{L}_{{r},h}^k$ and 
  \textbf{Store} $A^{k+1},B_{{r}}^{k+1}(\bar{\theta}_{{r}}),D_{{r}}^{k+1}(\bar{\theta}_{{r}}),F_{{r}h}^{k+1}(\bar{\theta}_{{r}}),\mathbf{L}_{{r},h}^{k+1}$ for all $  \bar{\theta}_{{r}}\in \Theta_{{r}}$ and for all ${r},h\in \mathcal{N}$\; 
{Update} stage index $k\leftarrow k+1$\;  
  }
 \caption{$K$-stage receding-horizon control for player $i\in \mathcal{N}$ of type $\theta_i\in \Theta_i$ 
 \label{algorithm:implentation}}
\end{algorithm}

The computational complexity of the belief matrix update in the line $15$ of Algorithm \ref{algorithm:implentation} is $O(N_0^N N)$. 
For any $\beta^k$, the term $\mathbf{W}^{0,k}(\beta^k)$ has computational complexity $O(N_0^N N)+O(N_0^3 N^2)$, which is determined by the belief matrix update and the matrix chain multiplication of $W_{ij}^{0,k}(\beta^k)$, respectively. 
Then, the computational complexity of $(\mathbf{W}^{0,k}(\beta^k))^{-1}$  and $\mathbf{W}^{1,k}(\beta^k)$ is $O(N_0^N N)+O( N_0^3 N^3)$ and $O(N_0^N N)+O(N_0^3N^2 )$, respectively. 
Given $\beta^k$ and $\theta_i$, the computational complexity of $S_i^k(\beta^k,\theta_i)$ in \eqref{eq:S} is $O(N_0^N N)+O(N_0^3N^3)+O(N_0^3N^2)+O(N_0N)=O(\max{ (N_0^N N, N_0^3N^3) })$, which hinges on the computational complexity of $\mathbf{M}_i^k(\beta^k,\theta_i)$ (or $(\mathbf{W}^{0,k}(\beta^k))^{-1}$), $\mathbf{W}^{1,k}(\beta^k)$, and the matrix chain multiplication in \eqref{eq:S}. 
Similarly,  $N_i^k(\beta^k,\theta_i)$ and $W^{2,k}(\beta^k)$ both have computational complexity of $O(N_0^N N)+O(N_0N)$. 
Therefore, player $i$'s temporal complexity at each stage $k\in \{0,1,\cdots,K-1\}$ is 
\begin{equation*}
    O((K-k)\cdot N_0N \cdot \max{ ( N_0^N N, N_0^3N^3) }).  
\end{equation*} 
The temporal complexity has the maximum value of $O(K\cdot \max{\{N_0^{N+1} N^2, N_0^4N^4\}})$ at the initial stage $k=0$ where each player has to predict the entire $K$ future stages to act optimally under the deception. 
Since the temporal complexity decreases as the real-time stage index $k$ increases, a player who can compute the equilibrium action within the required time at the initial stage $k=0$ is guaranteed to meet the real-time requirement in the following stages of interaction. 
If the number of types and agents are on the same scale, e.g., $N_0=N$, then $\lim_{N\rightarrow \infty} (N_0^{N+1} N^2)/(N_0^4N^4) \rightarrow \infty$ and the computation of belief matrix update plays a dominant role as each player keeps track of all players' beliefs to obtain the equilibrium action under deception. 
If $N_0 \ll N$, e.g., $N_0=N^{1/N}$, then $\lim_{N\rightarrow \infty} (N_0^{N+1} N^2)/(N_0^4N^4) \rightarrow 0$ and the inverse of $\mathbf{W}^{0,k}(\beta^k)$ becomes the most time-consuming operation due to the coupling in dynamics, costs, and cognition. 
Effective deception can prevent or delay other players from learning the deceiver's private type. We define the criterion of successful learning of the deceiver's type in Definition \ref{def:truth-revealing} and  $\epsilon$-deceviability and $\epsilon$-learnability in Definition \ref{def:Deceviability}. 

 \begin{definition}[\textbf{Stage of  Truth Revelation}]
 \label{def:truth-revealing}
Consider two players $i,j\in \mathcal{N}$ with type $\theta_i$ and $\theta_j$, respectively. 
Stage $k_{i,j}^{tr}\in \mathcal{K} \cup \{K+1\}$ is said to be player $i$'s truth-revealing stage with accuracy $\delta\in (0,1]$\footnote{Since the belief mismatch does not reduce to $0$ in finite stages with initial belief $l_i^0\in (0,1)$, the accuracy threshold $\delta\neq 0$.}  if it satisfies the following two conditions.
\begin{itemize}
    \item \textbf{The bounded mismatch condition}: 
    player $i$'s 
\textit{belief mismatch} remains less than $\delta$ after stage $k_{i,j}^{tr}\in \mathcal{K}$, i.e., 
\begin{equation}
\label{eq:tr1}
1-l^k_i(\theta_j |h^k,\theta_i) \leq  \delta, \forall k\geq k_{i,j}^{tr}.
\end{equation}
    \item \textbf{The first-hitting-time condition}: 
    $k_{i,j}^{tr}\in \mathcal{K}$ is the first stage satisfying \eqref{eq:tr1}, i.e., 
$
1-l^{k_{i,j}^{tr}-1}_i(\theta_j |h^{k_{i,j}^{tr}-1},\theta_i) > \delta, k_{i,j}^{tr}>1.  
$
\end{itemize}
If there does not exist  $k_{i,j}^{tr}\in \mathcal{K}$ that satisfies \eqref{eq:tr1}, we define $k_{i,j}^{tr}:=K+1$. If there are only two players $N=2$, we write $k_{i,j}^{tr}$ as $k_{i}^{tr}$ without ambiguity. 
\end{definition}

Due to deceivers' deceptive actions and the external noises, the belief sequence may be fluctuant; i.e., there can exist   $k< k_{i,j}^{tr}$ such that $1-l^k_i(\theta_j |h^k,\theta_i) \leq  \delta$. 
Thus, as shown in Definition \ref{def:truth-revealing}, a player should only claim a successful learning of other players' types if his belief mismatch remains less than $\delta$ for the remaining stages. 

\begin{definition}[\textbf{Deceviability and Learnability}]
\label{def:Deceviability}
Consider players $i,j\in \mathcal{N}$ with type $\theta_i$ and $\theta_j$, thresholds $\delta\in (0,1], \epsilon\in [0,1]$, and a given stage index $\tilde{k}\in \mathcal{K}\cup \{K+1\}$. 
Player $i$ is $\tilde{k}$-stage $\epsilon$-deceivable 
if the probability
$\Pr(k_{i,j}^{tr}<\tilde{k})$, \textcolor{black}{or equivalently $\Pr(l_i^{\tilde{k}}(\theta_j |x^{\tilde{k}},\theta_i)>1-\delta)$, is not greater than  $\epsilon$ for all $ l_i^0\in (0,1)$. }
If the above does not hold, player $j$'s type is said to be $\tilde{k}$-stage $\epsilon$-learnable by player $i$. 
\end{definition}

Since robot deception involves only a finite number of stages, it is essential that the deceived robot can learn the deceiver's type as quickly as possible so that he has sufficient stages to plan on and mitigate the deception impact from the previous stages. 
Therefore, the definition of  learnability, i.e., non-deceviability in Definition \ref{def:Deceviability}, not only requires the deceived player to be capable of learning the deceiver's private information, but also learning it in a desirable rate, i.e., within $\tilde{k}$ stage. 
Due to the external noise, $k_{i,j}^{tr}$ is a random variable. 
Thus, the definition of learnability requires $\Pr(k_{i,j}^{tr}<\tilde{k})> \epsilon$; i.e., player $i$ has a large probability to correctly learn the type of player $j$ before stage $\tilde{k}$. 

\section{Dynamic Target Protection under Deception} 
\label{sec:case study} 
We investigate a pursuit-evasion scenario that contains 
{
two UAVs with the decoupled linear time-invariant state dynamics, i.e., $A^k(\theta)=\mathbf{I}_{4}, \bar{B}_i^k(\theta_i)=[\tilde{B}_i(\theta_i), 0; 0, \tilde{B}_i(\theta_i)]\in \mathbb{R}^{2\times 2}, \forall k\in\mathcal{K}$.}  
We use `she' for UAV $1$, the pursuer, and `he' for UAV $2$, the evader.  
UAV $i$'s state $x_i^k:=[x^k_{i,x},x^k_{i,y}]'\in\mathbb{R}^{2\times 1}$ represents $i$'s location $(x^k_{i,x},x^k_{i,y})$ in the $2D$ space, and action $u_i^k=[u_{i,x}^k,u_{i,y}^k]\in \mathbb{R}^{2\times 1}$ affects $i$'s speed in $x$ and $y$ directions. 

UAV $2$ as the evader selects either the harbor in `Normandy' or `Calais' as his final location 
based on his type $\theta_2\in \{\theta_2^g,\theta_2^b\}$. 
He aims to reach `Normandy' located at $\gamma(\theta_2^g):=(x^g,y^g)$ in $K=40$ stages if his type is $\theta_2^g$, otherwise `Calais' located at $\gamma(\theta_2^b):=(x^b,y^b)$ if his type is $\theta_2^b$. 
UAV $1$ as the pursuer can make  interfering signals and aims to be close to UAV $2$ at the final stage to protect the  harbor targeted by the evader, i.e.,  
$g_1^k(x^k,u^k,\theta_1)= d_{12}^k(\theta_1)( (x^k_{2,y}-x^k_{1,y})^2+(x^k_{2,x}-x^k_{1,x})^2) 
 + f_{11}^k(\theta_1)((u_{1,x}^k)^2+(u_{1,y}^k)^2)-f_{12}^k(\theta_1)((u_{2,x}^k)^2+(u_{2,y}^k)^2), \forall k\in \mathcal{K}$,  
 where $ d_{12}^k(\theta_1)\in \mathbb{R}_{\geq 0}$ penalizes her distance from the evader at stage $k\in \mathcal{K}$, $f_{11}^k(\theta_1)\in \mathbb{R}_{\geq 0}$ prevents her from a high action cost, and $f_{12}^k(\theta_1)\in \mathbb{R}_{\geq 0}$ incites her opponent, i.e., the evader, to take costly actions. 
We classify UAV $1$ into two types, i.e., $\Theta_1= \{\theta_1^H,\theta_1^L\}$, based on her maneuverability represented by the value of $\tilde{B}_1(\theta_1)$. 
Given higher maneuverability $\tilde{B}_1(\theta_1^H)>\tilde{B}_1(\theta_1^L)$, the pursuer of type $\theta_1^H$ can obtain a higher speed under the same action $u_1^k$ and thus cover a longer distance. 

{The evader's goals of deceptive target reaching and pursuit evasion are incorporated into} the cost structure $ g_2^k(x^k,u^k,\theta_2)  = d_{2,b}^k(\theta_2) (   (x^k_{2,y}-y^b)^2+ (x^k_{2,x}-x^b)^2)+
      d_{2,g}^k(\theta_2) (   (x^k_{2,y}-y^g)^2+   (x^k_{2,x}-x^g)^2)- d_{21}^k(\theta_2) (  (x^k_{1,y}-x^k_{2,y})^2+ (x^k_{1,x}-x^k_{2,x})^2)
      +f_{22}^k(\theta_2)((u_{2,x}^k)^2+(u_{2,y}^k)^2)
 - f_{21}^k(\theta_2)((u_{1,x}^k)^2+(u_{1,y}^k)^2), \forall k\in \mathcal{K}$. 
{
Similar to the pursuer's cost parameters, $d_{21}^k(\theta_2)\in \mathbb{R}_{\geq 0}$ represents the evader's level of \textit{evasion determination} to keep a distance from the pursuer along the trajectory. 
The action costs of the evader and the pursuer are regulated by $f_{22}^k(\theta_2)\in \mathbb{R}_{\geq 0}$ and $f_{21}^k(\theta_2)\in \mathbb{R}_{\geq 0}$, respectively.  
The parameters $d_{2,b}^k(\theta_2)$ and $d_{2,g}^k(\theta_2)$ represent the evader's  attempt to head toward `Normandy' and `Calais', respectively, at stage $k\in \mathcal{K}$ under type $\theta_2\in \Theta_2$. 
We use the ratio $d_{2,g}^k(\theta_2)/d_{2,b}^k(\theta_2)$ to represent the evader's level of \textit{trajectory deception}.} 
Since the pursuer can learn the evader's type based on the real-time  observations of state $x_2^k$, the evader attempts to make his target $\epsilon_0$-ambiguous at all previous stages, 
 i.e., $|d_{2,b}^k(\theta_2)/d_{2,g}^k(\theta_2)-1|\leq \epsilon_0, \forall \theta_2, \forall k\neq K$, and reveal his true target only at the final stage, i.e., $d_{2,g}^K(\theta_2^b)=0$ and $d_{2,b}^K(\theta_2^g)=0$.  
The evader chooses a small $\epsilon_0\geq 0$ and achieves the maximum ambiguity when $\epsilon_0=0$. 
Two blue lines in Fig. \ref{fig:NAMEInprogressDecep} illustrate how the evader manages to remain ambiguous in a cost-effective manner from two different initial locations.   
Instead of keeping an equal distance to both potential targets, the evader heads toward the midpoint $((x^g+x^b)/2,(y^g+y^b)/2)$ at the early stages to confuse the pursuer. 
{However, the evader starts to head toward the true target at around half of $K$ stages rather than the last few stages so that he can reach the target with a moderate control cost $(u^{k}_2)' F^k_{22} (\theta_2)u^{k}_2$. Fig. \ref{fig:NAMEInprogressDecep} also shows that for a given initial location, the evader who adopts a higher level of \textit{trajectory deception} heads more toward the misleading target at the early stages.}  


In this case study, we suppose that the evader's true target is Calais and let $\theta_2^b$ be his \textit{true type} and $\theta_2^g$ be the \textit{misleading type}. 
{
The following two ratios capture the evader's tradeoff of being deceptive, effective, and evasive.}  
On one hand, the ratio $d_{2,b}^k(\theta_2^b)/d_{2,b}^K(\theta_2^b), k\neq K$, reflects the evader's tradeoff between applying deception along the trajectory and staying close to the true target at the final stage. 
{
Fig. \ref{fig:NAMEtargetreaching} shows that as the evader focuses more on a deceptive trajectory represented by a larger value of $d_{2,b}^k(\theta_2^b)/d_{2,b}^K(\theta_2^b), k\neq K$,} his trajectory remains ambiguous for longer stages while his final location is farther away from the true target. 
On the other hand, the ratio $d_{21}^k(\theta_2^b)/d_{2,b}^K(\theta_2^b), k\neq K$, reflects the evader's tradeoff between evasion and target-reaching. 
As the evader focuses more on keeping a distance from the pursuer along the trajectory, he takes a bigger detour and stays farther away from his true target at the final stage as shown in Fig. \ref{fig:NMAEevasion}. 
\begin{figure*}[t]
    \centering 
\begin{subfigure}{0.33\textwidth}
  \includegraphics[width=\linewidth]{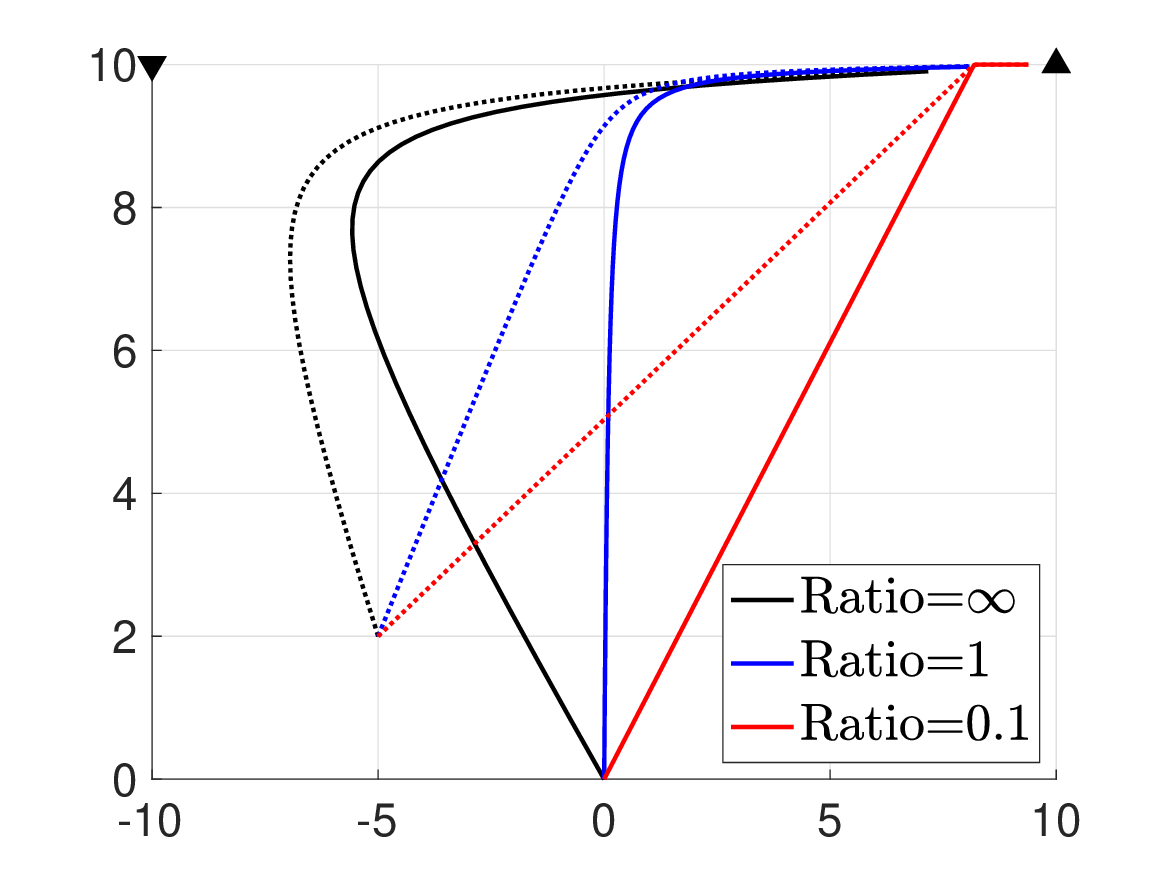}
  \caption{\label{fig:NAMEInprogressDecep} 
Ratio represents $d_{2,g}^k(\theta_2^b)/d_{2,b}^k(\theta_2^b)$.  
   }
\end{subfigure}\hfil 
\begin{subfigure}{0.33\textwidth}
  \includegraphics[width=\linewidth]{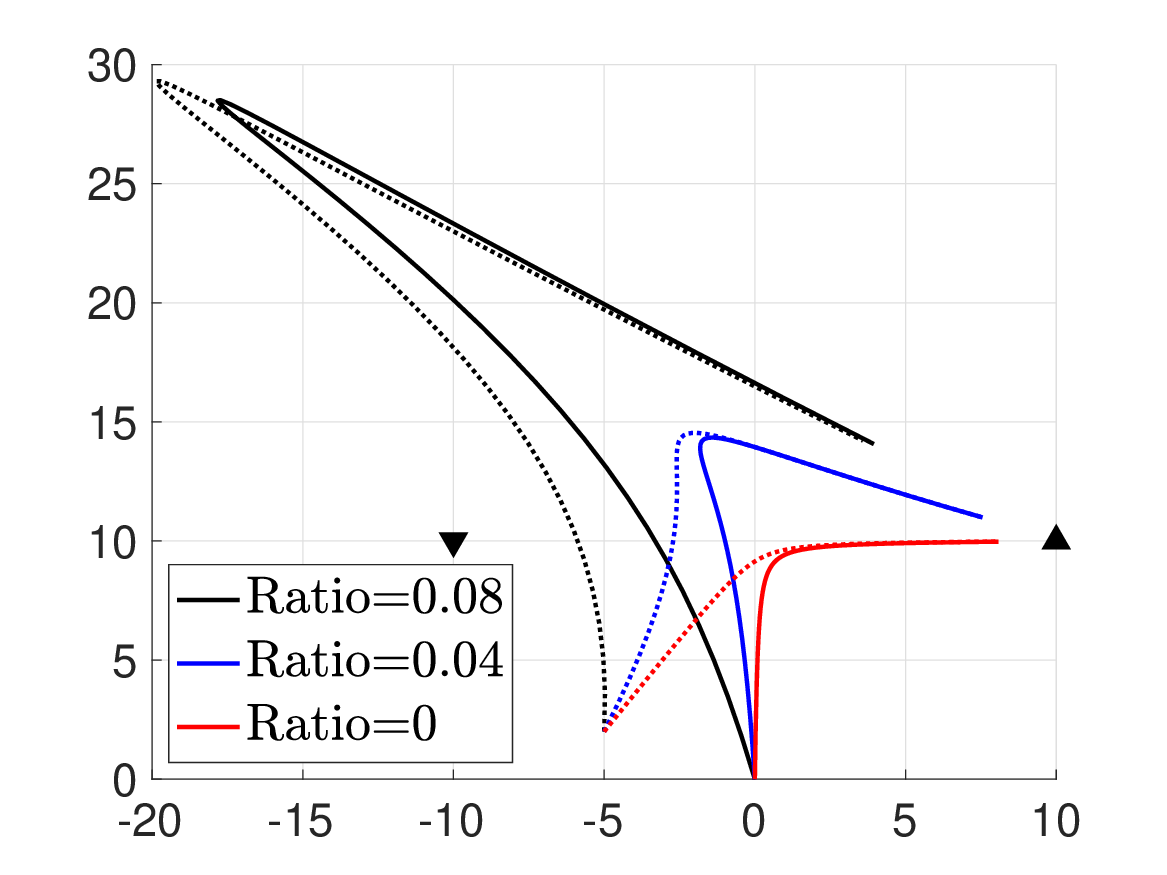}
  \caption{\label{fig:NAMEtargetreaching} 
  Ratio represents $d_{2,b}^k(\theta_2^b)/d_{2,b}^K(\theta_2^b)$. 
 }
\end{subfigure}\hfil 
\begin{subfigure}{0.33\textwidth}
 \includegraphics[width=\linewidth]{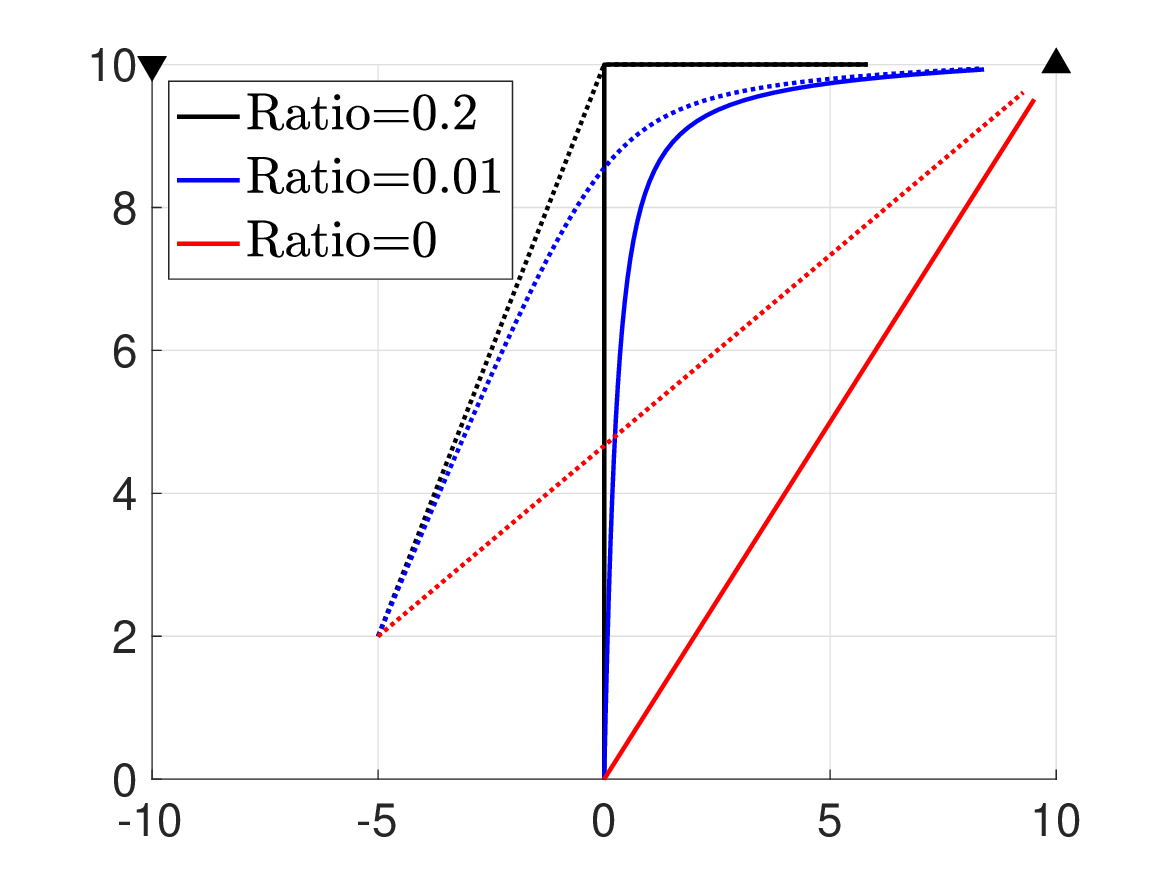}
   \caption{\label{fig:NMAEevasion} 
 Ratio represents $d_{21}^k(\theta_2^b)/d_{2,b}^K(\theta_2^b)$. 
 }
\end{subfigure}
\caption{
The evader's trajectories from $x^0_2=[0,0]$ and $x^0_2=[-5,2]$ in solid and the dashed lines, respectively. 
The black downward and upward triangles represent the location of Calais $(x^b,y^b)=(-10,10)$ and Normandy $(x^g,y^g)=(10,10)$, respectively.
\textcolor{black}{The ratios capture the evader's tradeoff of forming a deceptive trajectory,  reaching the true target, and evading the pursuit.} 
}
\label{fig:name}
\end{figure*}

Finally, we {
transform UAV $i$'s coupled cost $g_i^k$ into the matrix form given in Section \ref{sec:LQ}}, i.e., 
 $\hat{x}_1^k(\theta_1)=\mathbf{0}_{4,1}, \hat{f}_1^k(\hat{x}_1^k(\theta_1))=0, F^k_{ii}(\theta_1)=f^k_{ii}(\theta_1)\cdot \mathbf{I}_2, F^k_{ij}(\theta_1)=-f^k_{ij}(\theta_1) \cdot \mathbf{I}_2, j\neq i, 
  D_1^k(\theta_1)=d_{12}^k(\theta_1)\cdot  [1,0,-1,0;0,1,0,-1;-1,0,1,0;0,-1,0,1]$, 
\begin{equation*}
\resizebox{.99\hsize}{!}{
$D_2^k(\theta_2)=\begin{bmatrix}
-d_{21}^k & 0 &d_{21}^k  & 0 \\
0 & -d_{21}^k &0  & d_{21}^k  \\
d_{21}^k & 0 & d_{2,b}^k+d_{2,g}^k-d_{21}^k & 0  \\
0 & d_{21}^k &0 & d_{2,b}^k+d_{2,g}^k-d_{21}^k  \\
\end{bmatrix}$},
\end{equation*}
$\hat{x}_2^k(\theta_2)=\frac{1}{d_{2,b}^k+d_{2,g}^k}\cdot [d_{2,b}^k x^b+d_{2,g}^k x^g \: ; \:  d_{2,b}^k y^b+d_{2,g}^k y^g\: ; \: d_{2,b}^k x^b+d_{2,g}^k x^g \: ;\:  d_{2,b}^k y^b+d_{2,g}^k y^g]$,  $\hat{f}_2^k(\hat{x}_2^k(\theta_2))=\frac{ d_{2,b}^k  d_{2,g}^k  ((x^b-x^g)^2+ (y^b - y^g)^2)}{ d_{2,b}^k + d_{2,g}^k }$.

\subsection{Deceptive Evader with Decoupled Cost Structure}
We first investigate the scenario where the evader has a decoupled cost structure\footnote{
This paper has supplementary downloadable materials available at \url{http://ieeexplore.ieee.org}, provided by the authors. 
This includes a video demo of two UAVs' trajectories and belief updates under the decoupled structure.  
} defined in Definition \ref{def:decouple}, i.e., $d^k_{21}(\theta_2)=0, \allowbreak
\forall \theta_2\in \Theta_2,  \allowbreak
\forall k\in \mathcal{K}$. 
According to Corollary \ref{corollary: degeneration}, the evader's trajectory is then independent of the pursuer's action, type, and belief. 
Fig. \ref{fig:optimalcontrol} visualizes the pursuer's trajectories.  Although the pursuer only aims to be close to the evader at the final stage, she also takes proactive actions in the previous stages to be cost-efficient. 
If the pursuer knows the evader's type, then she can head toward the true target directly and will not be misled by the evader's trajectory ambiguity at the early stages \textcolor{black}{as illustrated by the black dashed line in Fig. \ref{fig:optimalcontrol}}. 
If the evader's type is private, then a larger initial belief mismatch $1-l_1^0(\theta_2^b|x^0,\theta_1^H)$ makes the pursuer head more toward the misleading target at the early stages \textcolor{black}{as illustrated by the three solid lines in Fig. \ref{fig:optimalcontrol}}. 
However, due to the pursuer's online learning, \textcolor{black}{which is compatible, efficient, and robust as shown in Section \ref{subsec:finitestageBelief}},  
she manages to approach the evader at the final stage regardless of her initial belief mismatch. 
Fig. \ref{fig:InitialBelief_optimalcontroltotal} shows the pursuer's $K$-stage belief variation. 
The evader's ambiguous trajectory results in belief fluctuations at the early stages, yet the pursuer can quickly reduce the belief mismatch when the evader starts to head toward the true target. 
After the pursuer has corrected her initial belief mismatch at around stage $k=16$, she can head toward the true target in the cost-efficient way; i.e, she attempts to keep a uniform linear motion under the external noise as shown in 
the upper right region of Fig. \ref{fig:optimalcontrol}. 

\begin{figure}
\centering
\includegraphics[width=1 \columnwidth]{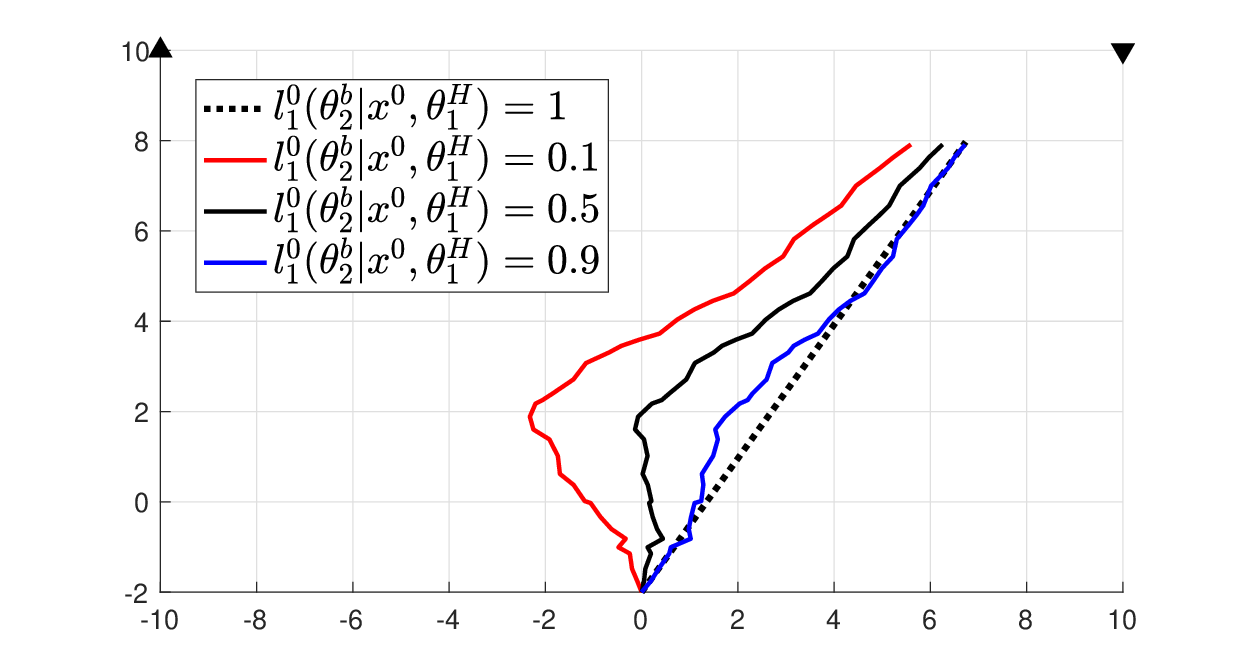}
\caption{The pursuer's trajectories under different initial beliefs. 
}
\label{fig:optimalcontrol}
\end{figure}

\begin{figure}
\centering
\includegraphics[width=1 \columnwidth]{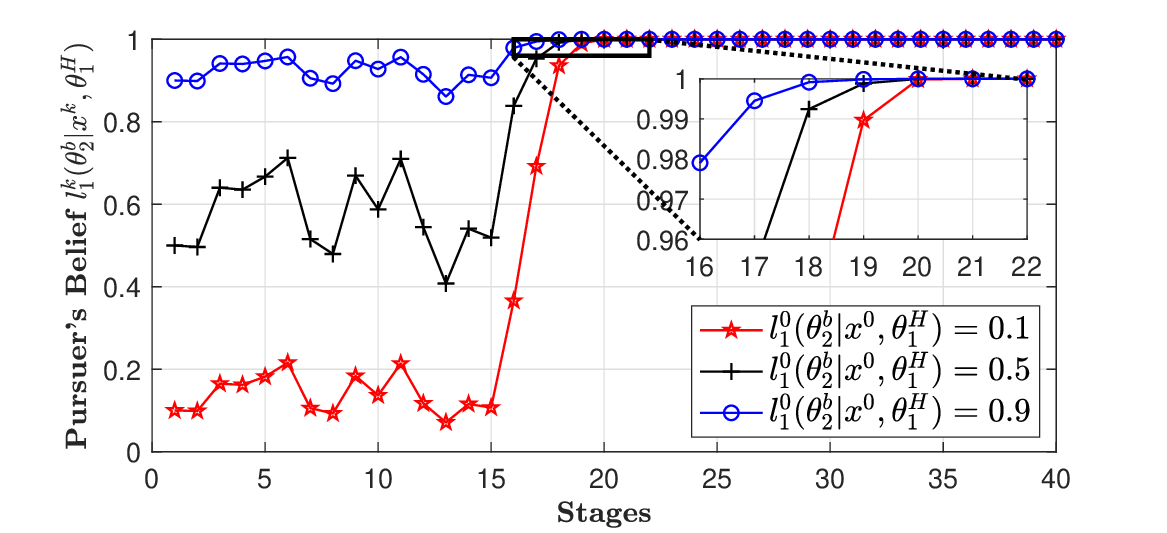}
\caption{ 
The pursuer's 
belief update over 
$K$ stages under three different initial beliefs and 
the same noise sequence $[w^k]_{k\in \mathcal{K}}$. The inset black box magnifies the selected area. 
}
\label{fig:InitialBelief_optimalcontroltotal}
\end{figure}

\subsubsection{Finite-Horizon Analysis of Bayesian Update}
\label{subsec:finitestageBelief}
\textcolor{black}{
In this subsection, we illustrate the compatibility, efficiency, and robustness of the finite-horizon Bayesian update in \eqref{eq:Bayesian} to reduce the initial belief mismatch. 
The pursuer is of high-maneuverability and the evader's true type is 
$\theta_2^b$. 
Define 
 the likelihood function of $\theta_2^b$ and $\theta_2^g$ as $a^k:=\Pr(x^{k+1}| \theta_2^b,x^k,\theta_1^H)$ and  $c^k:=\Pr(x^{k+1}| \theta_2^g,x^k,\theta_1^H)$, respectively. 
As $w^k\in \mathbb{R}^{n\times 1}$,  $a^k$ and $c^k$ are positive. 
 With an initial belief $l_1^0\in (0,1)$ and a finite likelihood ratio $e^k:=c^k/a^k \in (0,\infty)$, we can represent \eqref{eq:Bayesian} in the following form with three properties: 
\begin{equation*}
l_1^{k+1}=\frac{l_1^k \cdot a^k}{l_1^k \cdot  a^k+(1-l_1^k) \cdot c^k}
=\frac{1}{1+(\frac{1}{l_1^0} -1)\prod_{\bar{k}=0}^{k} e^{\bar{k}}}\in (0,1).  
\end{equation*} 
\begin{enumerate}
\item (\textbf{Compatibility}): For all $l_1^k\in (0,1)$, the belief update at stage $k$ is compatible to the evidence represented by the ratio $e^k$. 
In particular, if $e^k<1$, then $l_1^{k+1}>l_1^k$; if $e^k>1$, then $l_1^{k+1}<l_1^k$; if $e^k=1$, then $l_1^{k+1}=l_1^k$. 
\item (\textbf{Efficiency}): If the evidence of state observation $x^{k+1}$ indicates that the type is more likely to be the true type $\theta_2^b$, i.e., $e^k<1$, then the function $l_1^{k+1}/l_1^k=1/(l_1^k+(1-l_1^k)e^k)$  at stage $k$ is monotonically decreasing over $l_1^k$. 
If the evidence indicates that the type is more likely to be the misleading type $\theta_2^g$, i.e., $e^k>1$, then the function  $l_1^{k+1}/l_1^k$ is  monotonically increasing  over $l_1^k$. 
\item (\textbf{Robustness}): The order of the evidence sequence $e^{\bar{k}}, \bar{k}=0,\cdots,k$, has no impact on the belief $l_1^{k+1}$. 
\end{enumerate}
 Property one shows that although the external noise can result in the 
 fluctuations of the belief update, the belief mismatch, i.e., $1-l_1^k$, will decrease when $e^k<1$, regardless of the prior belief $l_1^k\in (0,1)$. 
 Property two shows the efficiency of the belief update. 
 The belief changes more under a larger belief mismatch, which results in a quick correction.  
Property three shows the robustness of the belief update. 
The erroneous belief update caused by a heavy noise can be corrected in the later stages when the noise fades.}

\subsubsection{Comparison with Heuristic Policies}
We compare the proposed pursuer's \textcolor{black}{control policy} with two heuristic ones to demonstrate its efficacy in \textcolor{black}{counter-deception}\footnote{
The supplementary materials include a video demo that compares the proposed  policy's trajectory and performance  with two heuristic policies.
}. 
The first heuristic policy is to repeat the attacker's trajectory with a one-stage delay; i.e.,  the pursuer applies the action so that $x_1^{k+1}=x_2^k, \forall k\in \mathcal{K}\setminus \{K\}$. 
The pursuer does not need to apply Bayesian learning and we name this policy as \textit{direct following}. 
The second heuristic policy for the pursuer is to stay at the initial location until her truth-revealing stage $k_1^{tr}$ and then head toward the evader's expected final-stage location in the remaining stages. 
The second policy is \textit{conservative} because the pursuer does not take proactive actions until she identifies the evader's type. 

Let player $i$'s \textit{ex-post cumulative cost} $\hat{V}_i^{0:k}:=\sum_{h=0}^k g_i^{h}, \forall k\in \mathcal{K}$, be a real-time evaluation of the online algorithm. 
Although a pursuer under both heuristic policies manages to stay close to the evader at the final stage, Fig. \ref{fig:bechmark} shows that 
both heuristic policies are more costly than the proposed equilibrium strategy in the long run. 
\begin{figure*}[htb]
    \centering 
\begin{subfigure}{0.3\textwidth}
  \includegraphics[width=\linewidth]{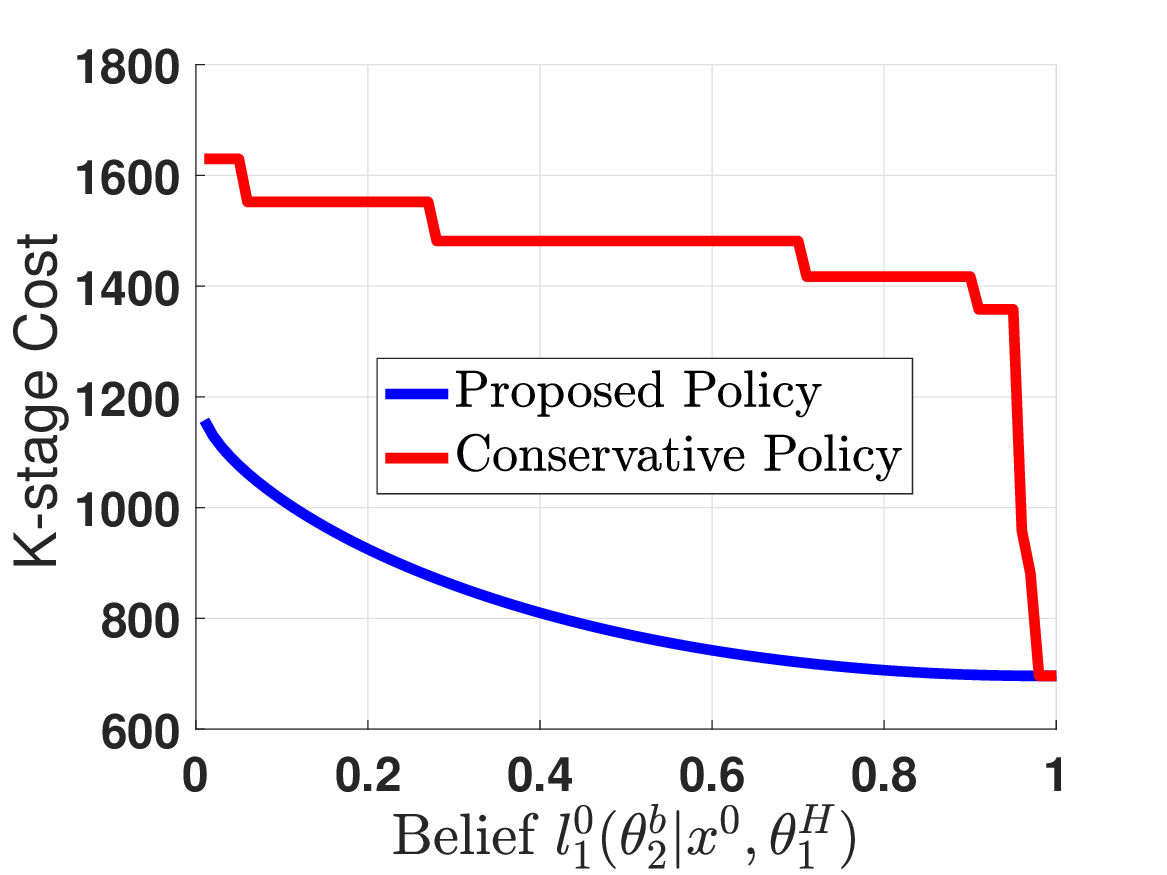}
  \caption{\label{fig:BenchConservative} 
  The $K$-stage cumulative cost $\hat{V}_i^{0:K}$ versus different initial beliefs. }
\end{subfigure}\hfil 
\begin{subfigure}{0.3\textwidth}
  \includegraphics[width=\linewidth]{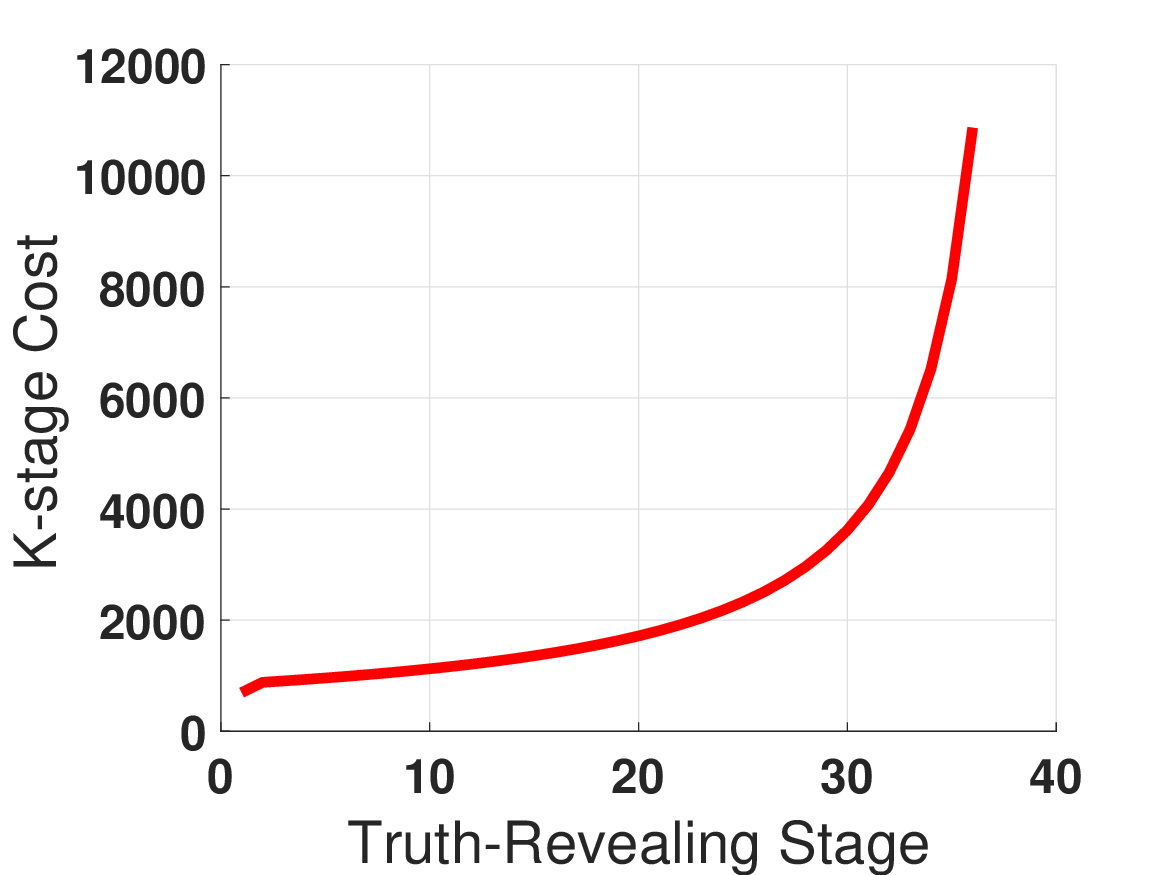}
  \caption{\label{fig:BenchExp} 
  The $K$-stage cumulative cost $\hat{V}_i^{0:K}$ versus  $k_1^{tr}$ under the \textit{conservative policy}. 
 }
\end{subfigure}\hfil 
\begin{subfigure}{0.3\textwidth}
 \includegraphics[width=\linewidth]{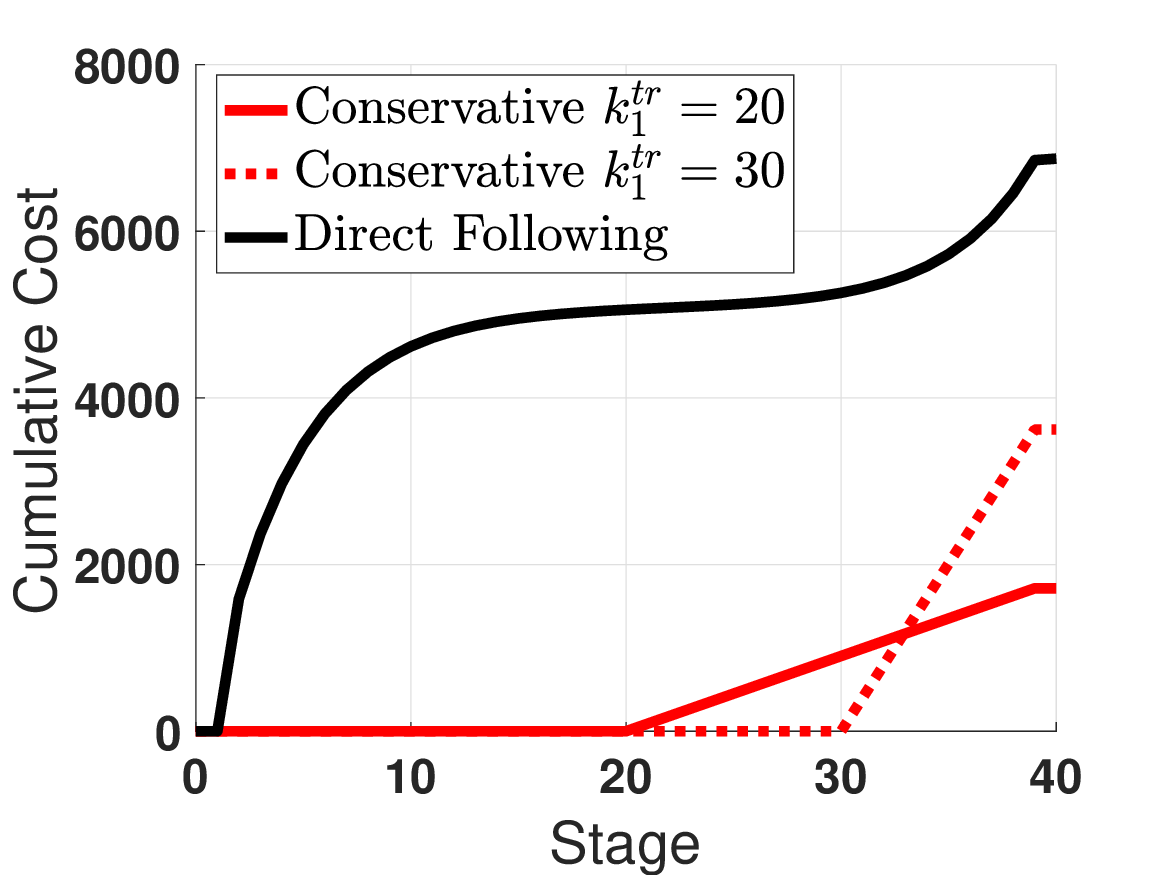}
   \caption{\label{fig:Bench3c} 
  The accumulation of the pursuer's cost $\hat{V}_i^{0:k},\forall k\in\mathcal{K}$, along with stages. 
   }
\end{subfigure}
\caption{
The pursuer's ex-post cumulative cost  under two heuristic policies and the proposed policy. 
\label{fig:bechmark}
}
\end{figure*}
The {conservative policy} avoids potential trajectory deviations under deception but results in less planning stages for the pursuer to achieve the capture goal. 
We visualize the accumulation of the pursuer's cost in Fig. \ref{fig:Bench3c}. The red lines show that the pursuer who adopts the {conservative policy} spends no action costs before the truth-revealing stage $k_1^{tr}$, i.e., $(u^{k}_1)' F^k_{11} (\theta_1)u^{k}_1=0, \forall k\leq k_1^{tr}$,  but huge costs in the remaining stages to fulfill her capture goal. The total cumulative cost $\hat{V}_i^{0:K}$ at the final stage increases exponentially with the value of $k_1^{tr}$ as shown in Fig. \ref{fig:BenchExp}. 
The black line in Fig. \ref{fig:Bench3c} illustrates the accumulation of $\hat{V}_i^{0:k}$ when the pursuer direct follows the evader's trajectory. Only under extreme deception scenarios where $k_1^{tr}>34$, the {direct following} policy results in a lower cost than the {conservative policy} does. 
Since the initial belief $l_1^0$ affects both the truth-revealing stage and the proposed policy, we plot $\hat{V}_i^{0:K}$ versus $l_1^0$ under the {conservative policy} and the proposed policy in Fig. \ref{fig:BenchConservative}. 
When there is no belief mismatch $l_1^0(\theta_2^b|x^0,\theta_1^H)=1$, we have $k_1^{tr}=1$ and the conservative policy is equivalent to the proposed policy. 
As the belief mismatch increases, the  cost $\hat{V}_i^{0:K}$ under the proposed policy (resp. the conservative policy) increases due to the larger deviation along the $x$-axis (resp. the larger $k_1^{tr}$). 
The proposed policy always results in a lower  cost $\hat{V}_i^{0:K}$ than the conservative policy does. 
The results in Fig. \ref{fig:bechmark} lead to the following two principles for the pursuer to behave under deception. 
First, Bayesian learning is a more effective countermeasure than the direct following of the evader's deceptive trajectory. 
Second, if learning the evader's type takes a long time, the pursuer is better to act proactively based on her current belief than to delay actions until the truth-revealing stage.

\subsection{Dynamic Game for Deception and Counter-Deception}
In this section, the evader has a coupled cost\footnote{
A video demo of two UAVs' real-time trajectories and belief updates under the coupled structure is included in the supplementary materials. 
}  defined in Definition \ref{def:decouple} and the level of \textit{evasion determination} increases with a constant rate $\alpha>0$; i.e., $d^k_{21}(\theta_2)=\alpha k, \forall \theta_2\in \Theta_2, \forall k\in \mathcal{K}$. 
The evader deceives the pursuer by hiding his true target. 
The pursuer can adopt the following two countermeasures to reduce her cost under the evader's deception.  
Section \ref{sec:public type} investigates the effectiveness of adaptive learning. We find that the pursuer manages to approach the true target at the final stage by updating her belief and taking actions accordingly based on the real-time trajectory observation. 
 Section \ref{sec:Dep for CounterDep} further allows the pursuer to introduce additional deception, i.e., obfuscate her maneuverability, to counteract the evader's information advantage and his deception impact.  

\subsubsection{Pursuer with a Public Type}
\label{sec:public type}
When the pursuer's type is \textit{common knowledge}, we plot both UAVs' trajectories under two initial beliefs and two types of pursuers  
 in Fig. \ref{fig:plot_gamecontrol}. 
\begin{figure}
\centering
\includegraphics[width=1.05 \columnwidth]{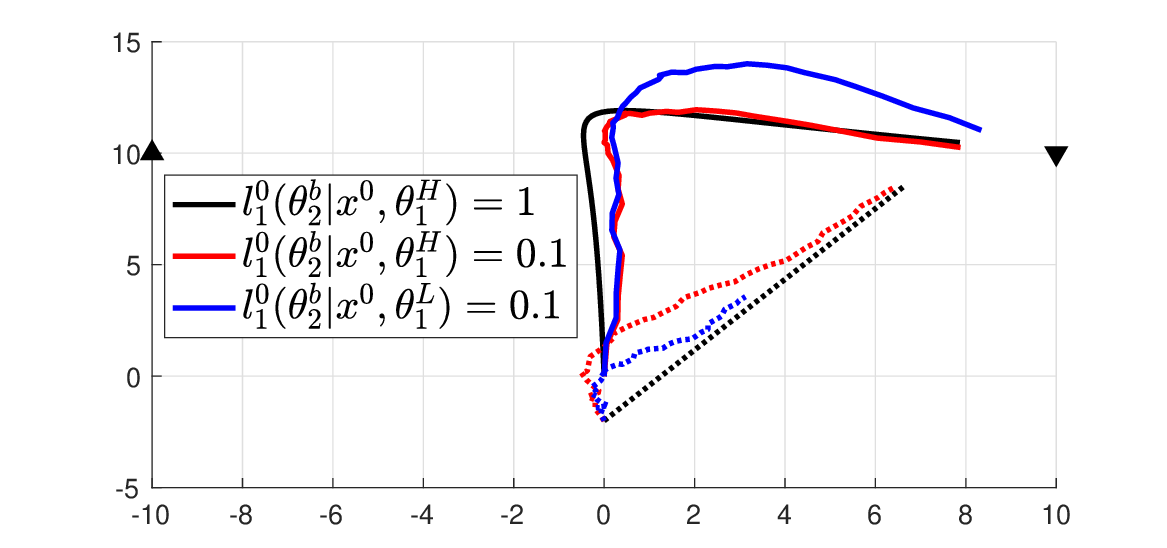}
\caption{ 
The $K$-stage trajectory of the evader and the pursuer in solid and dashed lines, respectively. 
If the evader's type is \textit{common knowledge} and the pursuer is of high-maneuverability, we represent their noise-free trajectories in black. 
If the evader's type is private and the pursuer's initial belief mismatch is $0.9$, two UAVs' trajectories are in red (resp. blue) when the pursuer's maneuverability is high (resp. low). 
}
\label{fig:plot_gamecontrol}
\end{figure}
The solid lines show that  the evader with the coupled cost detours to stay further from the pursuer. 
The initial belief mismatch causes a deviation along the $x$-axis for both high- and  low-maneuverability pursuers as shown in red and blue, respectively. 
However, the deviation has a smaller magnitude and lasts shorter than the one represented by the red line in Fig. \ref{fig:optimalcontrol} due to the coupled cost structure of the evader. 
The  pursuer with a high maneuverability stays closer to the evader at the final stage. 

\subsubsection{Deception to Counteract Deception}
 \label{sec:Dep for CounterDep}
When the pursuer's type is also private, Fig. \ref{fig:deceptionforcounter} shows that she can manipulate the evader's initial belief $l_2^0$ to obtain a smaller $k_1^{tr}$ and a belief update with less fluctuation. 
The red line with stars is the same as the one in Fig. \ref{fig:InitialBelief_optimalcontroltotal}. 
It shows that the pursuer's belief learning is slower and fluctuates more when she interacts with the evader who has a decoupled cost. The reason is that her manipulation of the initial belief $l_2^0$ does not affect the evader's decision making as shown in Corollary \ref{corollary: degeneration}. 
\begin{figure}[t]
    \centering 
\begin{subfigure}{0.5\textwidth}
  \includegraphics[width=\linewidth]{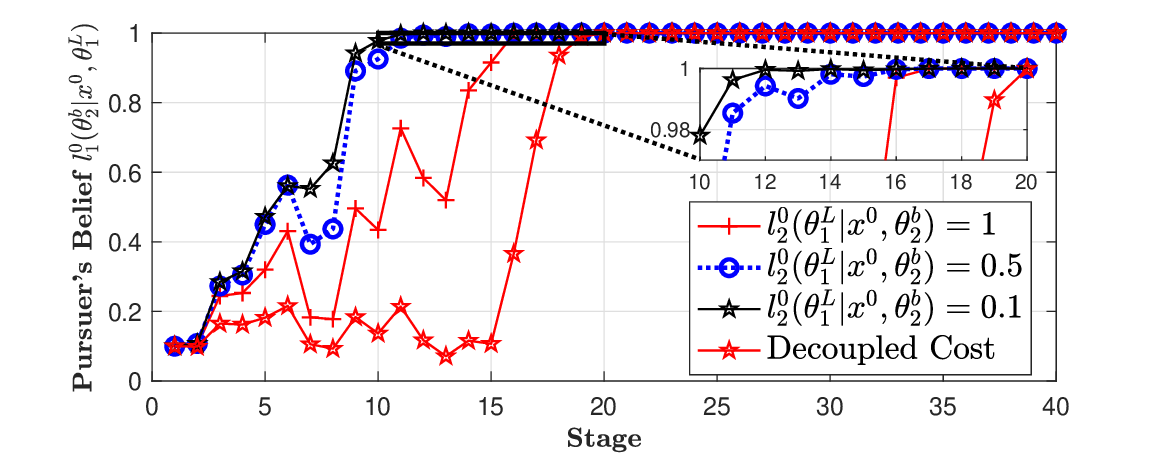}
  \caption{\label{fig:decerivedBelief_LowCap} 
Low-maneuverability pursuer's belief update. }
\end{subfigure}\hfil 
\begin{subfigure}{0.5\textwidth}
  \includegraphics[width=\linewidth]{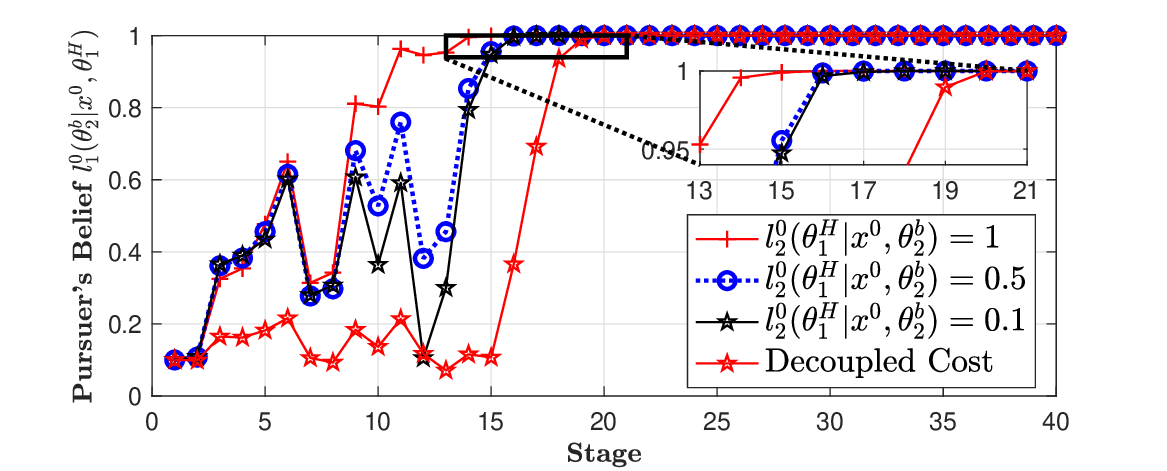}
  \caption{\label{fig:decerivedBelief} 
High-maneuverability pursuer's belief update. 
 }
\end{subfigure}\hfil 
\caption{
The pursuer's belief update over $K$ stages with the same initial belief  $l_1^0(\theta_2^b|x^0,\theta_1)=0.1$. The inset black box magnifies the selected area.  
\label{fig:deceptionforcounter}
}
\end{figure}
A comparison between Fig. \ref{fig:decerivedBelief_LowCap} and Fig. \ref{fig:decerivedBelief} shows that it is beneficial for a low-maneuverability pursuer to disguise as a high-maneuverability pursuer but not vice versa. Thus, introducing additional deception to counteract existing deception is not always effective.  


\subsection{Multi-Dimensional Deception Metrics}
The impact of the evader's deception can be measured by metrics such as the \textit{endpoint distance} $x_2^{fd}:=||x_2^K-\gamma(\theta_2)||_2$ between the evader and the true target, the \textit{endpoint distance} $x_1^{fd}:=||x_2^K-x_1^K||_2$ between two UAVs, both UAVs' truth-revealing stages $k_i^{tr}$, and their ex-post cumulative costs $\hat{V}_i^{0:k}, \forall k\in \mathcal{K}$. 
In this pursuit-evasion case study, we define $\epsilon$-reachability and $\epsilon$-capturability in Definition \ref{def:reach and cap}. 
Although $x_i^{fd}, \forall i\in\{1,2\}$, is a random variable,  
we can obtain a good estimate of the reachability and capturability due to the negligible  variance of $x_i^{fd}$ as shown in Fig. \ref{fig:DMfinalDis} and Fig. \ref{fig:ManeuverabilityDistan}. 
\begin{definition}[\textbf{Reachability and Capturability}]
\label{def:reach and cap}
Consider the proposed pursuit-evasion scenario with a given $\epsilon\geq 0$, a threshold $\bar{x}^{fd}\geq 0$, and all initial beliefs $l_i^0\in (0,1)$. The target is said to be $\epsilon$-reachable if $\Pr(x_2^{fd}\geq \bar{x}^{fd})\leq \epsilon$. The evader is said to be  $\epsilon$-capturable if $\Pr(x_1^{fd}\geq \bar{x}^{fd})\leq \epsilon$. 
\end{definition}

In Section \ref{sec:Manipulation}, we investigate how the evader can manipulate the pursuer's initial belief $l_1^0(\theta_2^b |x^0,\theta_1^H)$ to influence the deception. 
In Section \ref{sec:Maneuverability}, we investigate how the pursuer's maneuverability plays a role in deception.  
In both sections, the evader has a coupled cost structure. 
The pursuer  either applies the Bayesian update or not, which is denoted by blue and red lines, respectively, in both Fig. \ref{fig:DM} and Fig. \ref{fig:Maneuverability}. 
\textcolor{black}{In Section \ref{subsect:DDandPoD}, we study other metrics, such as deceivability, distinguishability, and PoD.}  

\subsubsection{The Impact of the Evader's Belief Manipulation}
\label{sec:Manipulation}
Both UAVs determine their initial beliefs  based on the intelligence collected before their interactions. 
By falsifying the pursuer's  intelligence, the evader can manipulate the pursuer's initial belief $l_1^0$ and further influence the deception as shown in Fig. \ref{fig:DM}. 
\begin{figure}
  \centering 
\begin{subfigure}{0.22\textwidth}
 \centering 
  \includegraphics[width=1.06\linewidth]{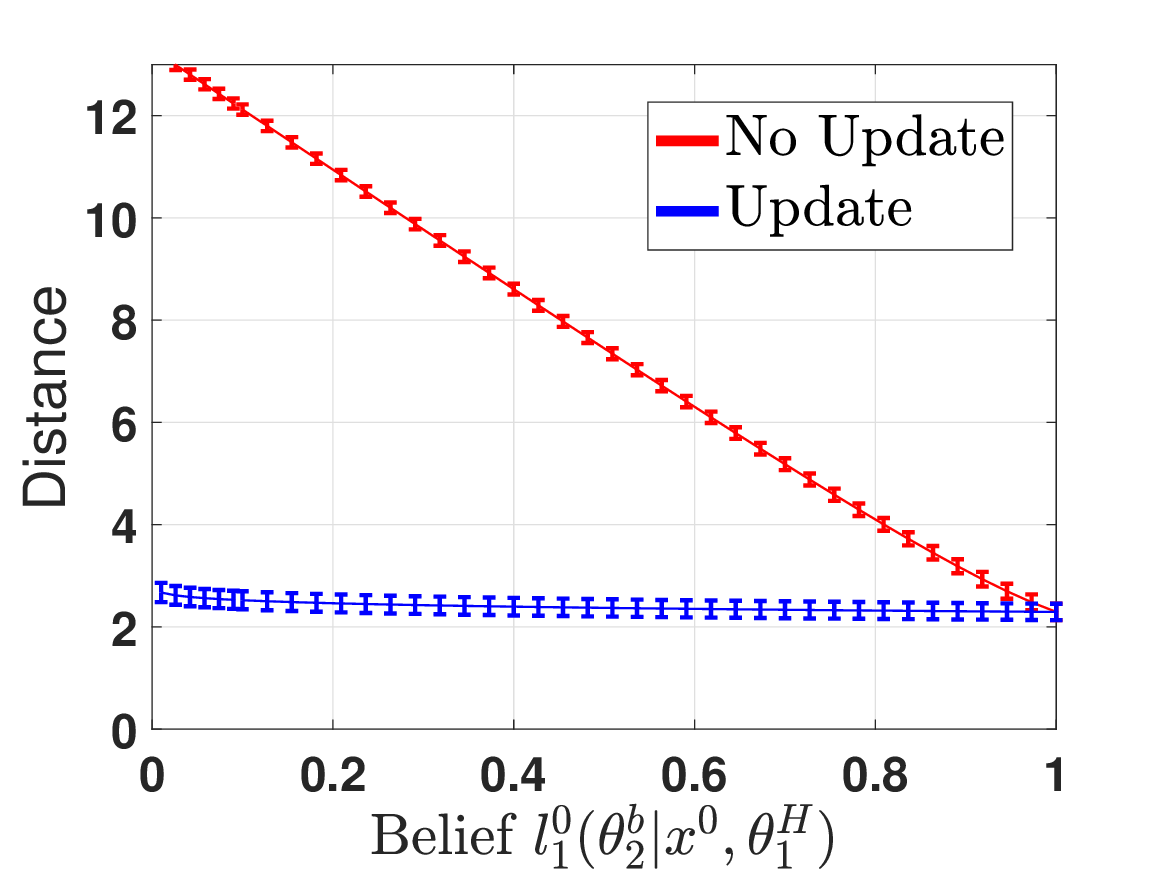}
  \caption{\label{fig:DMfinalDis} 
  Distance $x_1^{fd}$ with its variance magnified by $100$ times. }
\end{subfigure} \hfil 
\begin{subfigure}{0.22\textwidth}
  \includegraphics[width=1.065\linewidth]{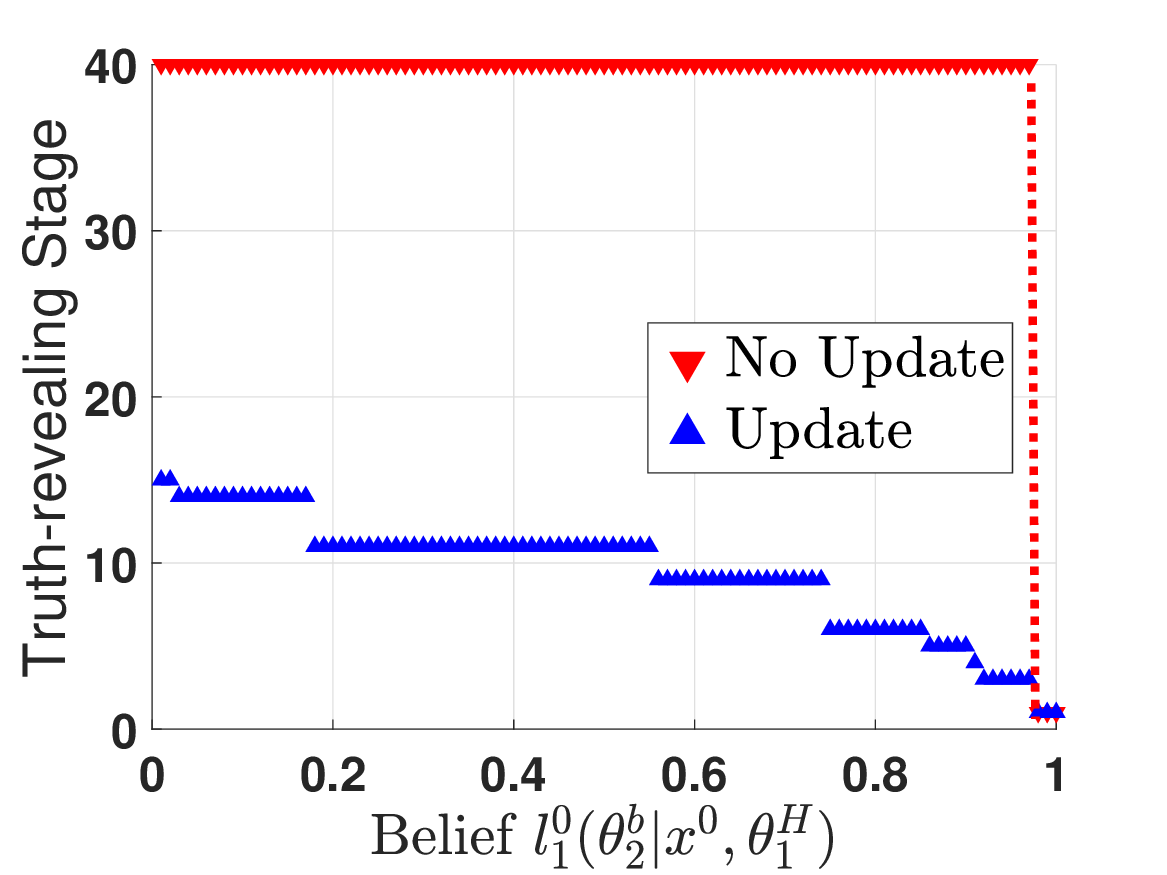}
  \caption{\label{fig:DMtruthrevealing} 
A realization of the pursuer's truth-revealing stage ${k}_1^{tr}$. 
 }
\end{subfigure}\hfil 
\medskip
\begin{subfigure}{0.22\textwidth}
 \includegraphics[width=1.06\linewidth]{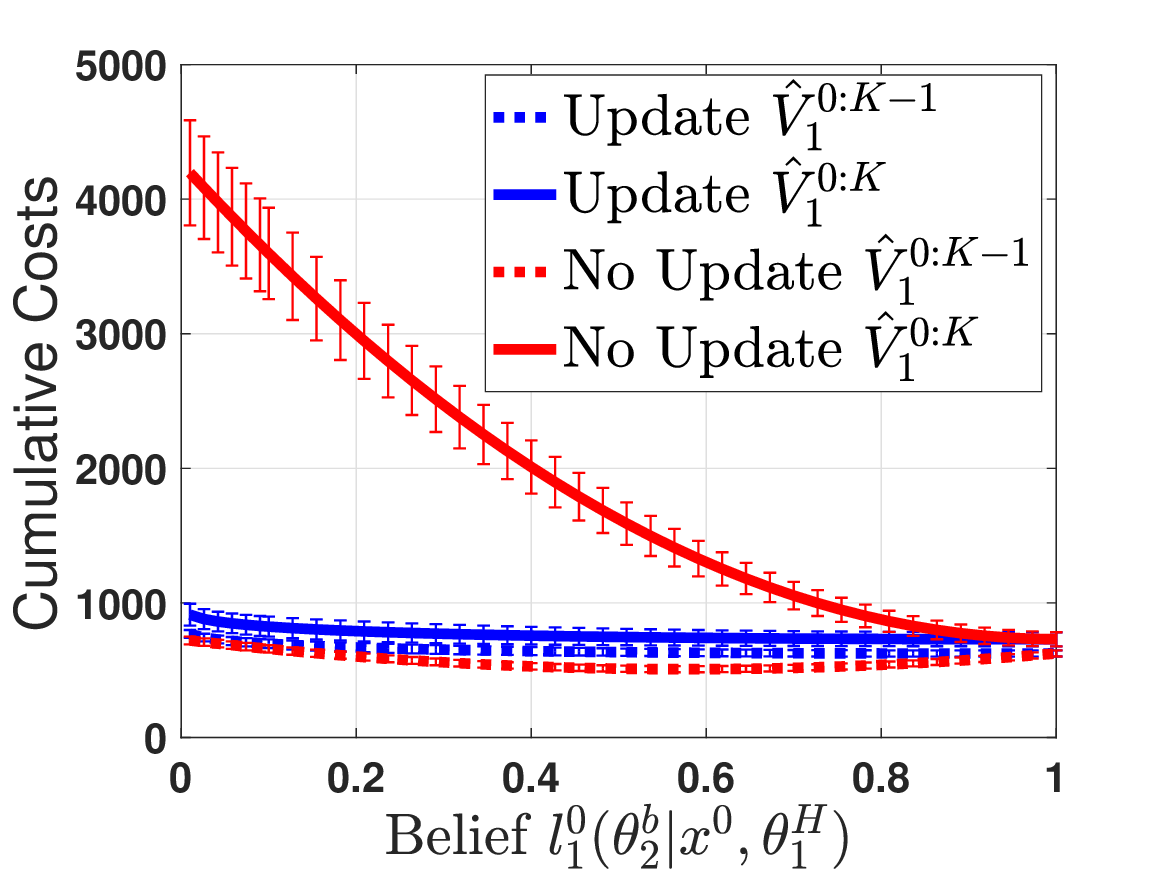}
   \caption{\label{fig:DMprogressiveCost} 
The  costs $\hat{V}_1^{0:K-1}$ and $\hat{V}_1^{0:K}$ of the pursuer under type $\theta_1^H$. 
 }
\end{subfigure}\hfil 
\begin{subfigure}{0.22\textwidth}
  \includegraphics[width=1.06\linewidth]{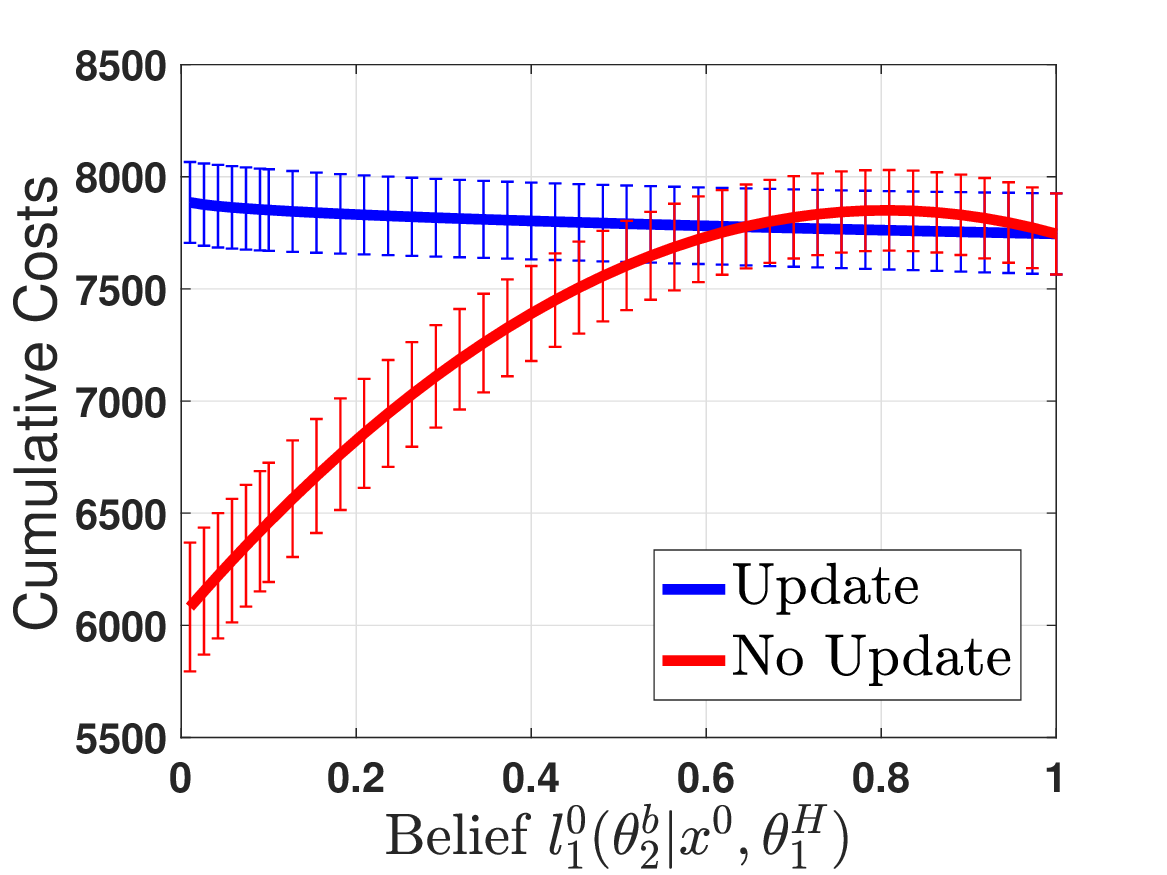}
  \caption{\label{fig:DMratio} 
The evader's $K$-stage ex-post cumulative cost $\hat{V}_2^{0:K}$. 
 }
\end{subfigure}
\caption{
The influence of the initial belief mismatch on deception.
Error bars represent variances of the random variables. 
}
\label{fig:DM}
\end{figure}
In the  $x$-axis, an initial belief $l_1^0(\theta_2^b |x^0,\theta_1^H)$ closer to $1$ indicates a smaller belief mismatch. 
 Fig. \ref{fig:DMfinalDis} shows that the pursuer's distance to the evader at the final stage decreases as the belief mismatch decreases regardless of the existence of Bayesian learning. 
However, the initial belief manipulation has a much less influence on the endpoint distance $x_1^{fd}$ when Bayesian learning is applied. 
Fig. \ref{fig:DMtruthrevealing} shows that for each realization of the noise sequence $w^k$, the pursuer's truth-revealing stage  steps down as the belief mismatch decreases when Bayesian update is applied. 
Fig. \ref{fig:DMprogressiveCost} illustrates the pursuer's ex-post cumulative cost $\hat{V}_1^{0:K}$ and $\hat{V}_1^{0:K-1}$ at the last and the second last stage, respectively. 
Without Bayesian update, 
the evader's deception significantly increases the pursuer's cost at the second last stage due to the large endpoint distance $x_1^{fd}$. 
The red lines show that the cost increase is higher under a larger belief mismatch. 
 Fig. \ref{fig:DMratio} illustrates the evader's ex-post cumulative cost at the last stage. If the pursuer does not apply Bayesian learning, then the evader can decrease his cost by increasing the pursuer's belief mismatch. 
 If the pursuer applies Bayesian learning, then the evader's cost increases slightly if the pursuer's belief mismatch is increased. 
When the belief mismatch is small (i.e., $1-l_1^0\in (0,0.35)$), we observe a win-win situation; i.e., Bayesian learning not only reduces the pursuer's ex-post cumulative cost, but also the evader's. 

\subsubsection{The Impact of the Pursuer's Maneuverability}
\label{sec:Maneuverability}
The pursuer's maneuverability can also affect deception as shown in Fig. \ref{fig:Maneuverability}. 
\begin{figure}
    \centering 
\begin{subfigure}{0.22\textwidth}
  \includegraphics[width=1.06\linewidth]{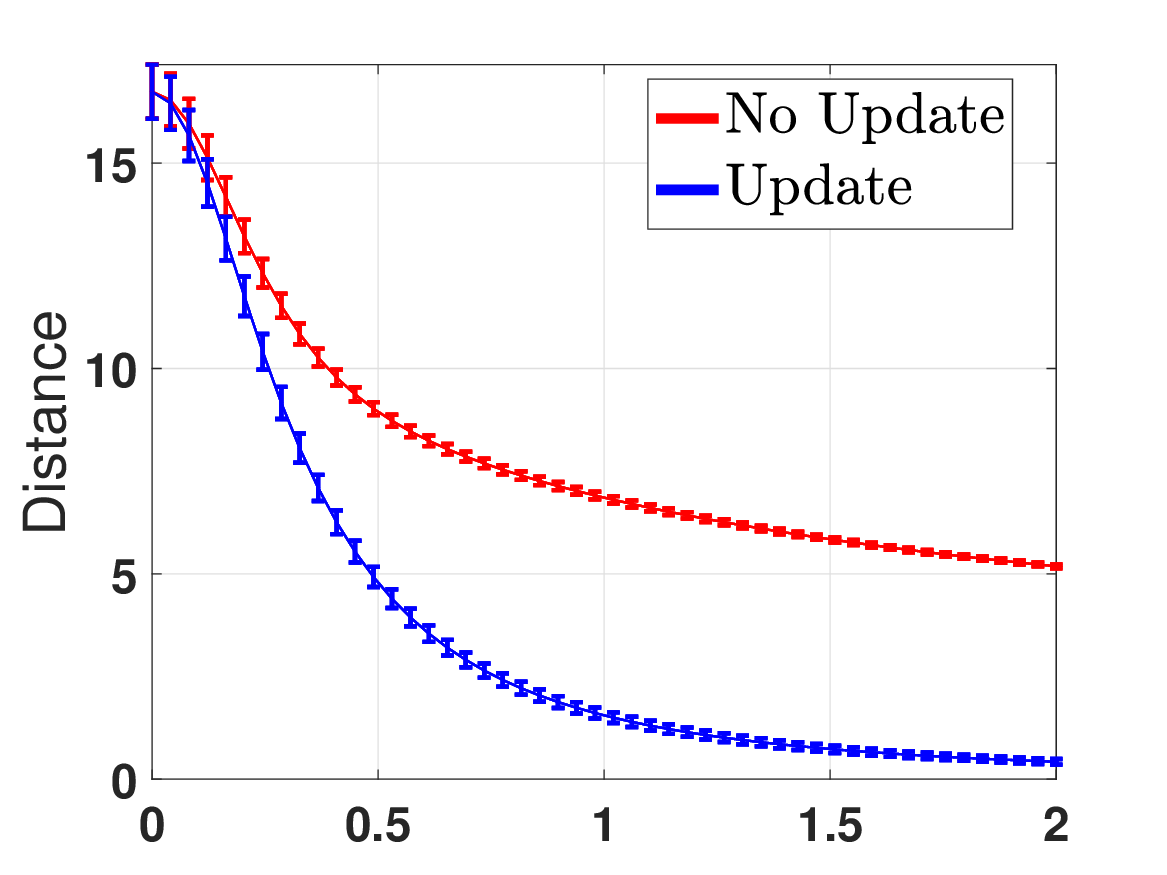}
  \caption{\label{fig:ManeuverabilityDistan} 
Distance $x_1^{fd}$ with its variance magnified by $100$ times. }
\end{subfigure}\hfil 
\begin{subfigure}{0.22\textwidth}
  \includegraphics[width=1.06\linewidth]{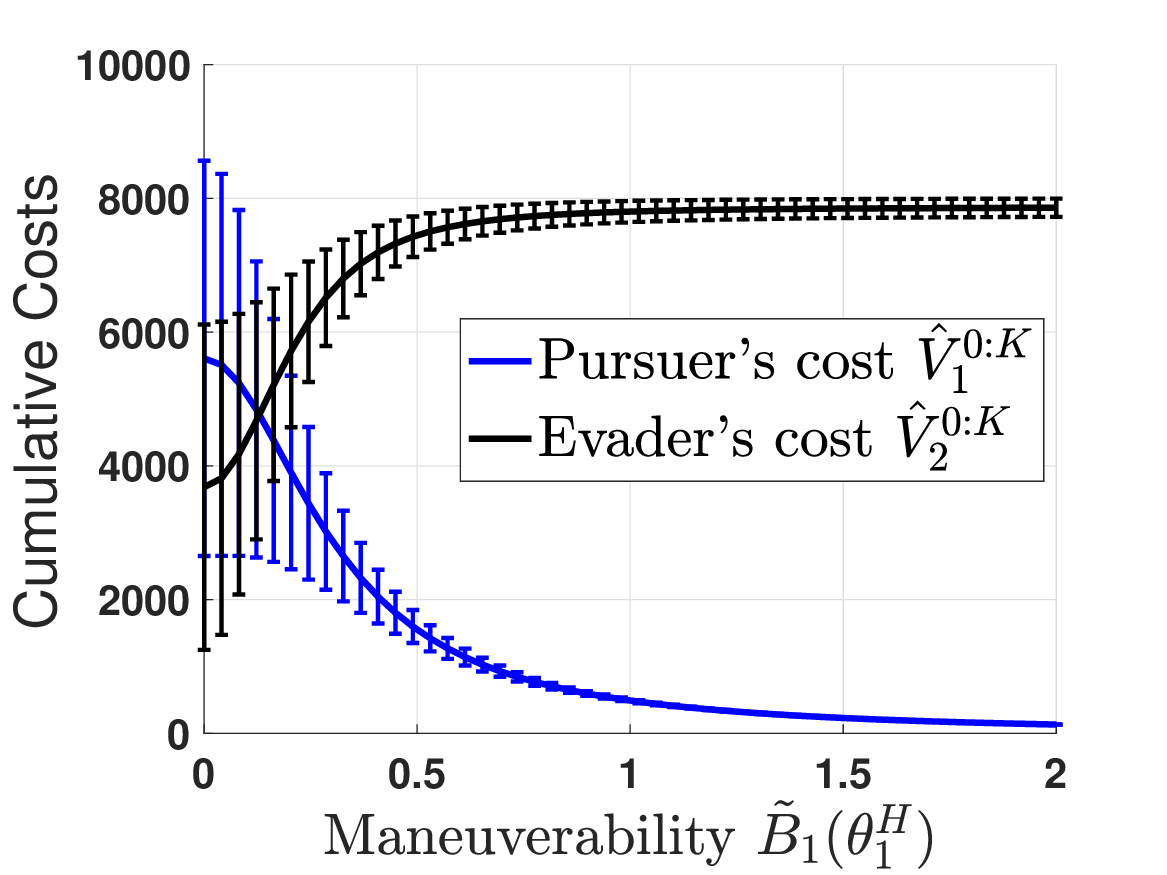}
  \caption{\label{fig:Maneuverabilityvalue} 
Two UAVs' $K$-stage costs $\hat{V}_1^{0:K}$ and $\hat{V}_2^{0:K}$. 
 }
\end{subfigure}\hfil 
\caption{
The influence of the pursuer's maneuverability on deception. Error bars represent variances of the random variables. 
\label{fig:Maneuverability}
}
\end{figure}
The pursuer has an initial belief $l_1^0(\theta_2^b |x^0,\theta_1^H)=0.5$ and the evader knows the pursuer's type. 
Fig. \ref{fig:ManeuverabilityDistan} illustrates that the pursuer can 
exponentially decrease her distance to the evader at the final stage as her maneuverability increases. 
Fig. \ref{fig:Maneuverabilityvalue} demonstrates that the  maneuverability increase can decrease and increase the pursuer's and the evader's ex-post cumulative \textcolor{black}{costs} at the final stage, respectively. 
The variance grows as maneuverability decreases because the pursuer's trajectory will become largely affected by the external noise. 
In both figures, we observe the phenomenon of \textcolor{black}{the} \textit{marginal effect}; i.e., the change rates of both the endpoint distance $x_1^{fd}$ and the cost $\hat{V}_i^{0:K}$ decrease as the maneuverability increases. 
Thus, we conclude that higher maneuverability can improve the pursuer's performance under the evader's deception as measured by the distance $x_1^{fd}$ and the cost $\hat{V}_1^{0:K}$. 
Moreover, the improvement rate is higher with low maneuverability. 

\subsubsection{Deceivability, Distinguishability, and PoD}
\label{subsect:DDandPoD}
Deceivability defined in Definition \ref{def:Deceviability} is highly related to the distinguishblity among different types. 
In this case study, a larger distance between targets, i.e.,  $||\gamma(\theta_2^g)-\gamma(\theta_2^b)||_2$, makes it easier for the pursuer to distinguish \textcolor{black}{between} evaders of type $\theta_2^b$ and  \textcolor{black}{type} $\theta_2^g$. 
A larger maneuverability difference $|\tilde{B}_1(\theta_1^H)-\tilde{B}_1(\theta_1^L)|$ makes it easier for the evader to distinguish \textcolor{black}{between} pursuers of type $\theta_1^H$ and  \textcolor{black}{type} $\theta_1^L$. 
We visualize two UAVs' truth-revealing stages $k_i^{tr}$ versus the distance between targets and the maneuverability difference in Fig. \ref{fig:deceivability}. The evader has a coupled cost and both players' initial belief \textcolor{black}{mismatches are} $0.5$. The dashed black line indicates $\tilde{B}_1(\theta_1^L)=0.3$. 
\begin{figure}
\centering
\includegraphics[width=1 \columnwidth]{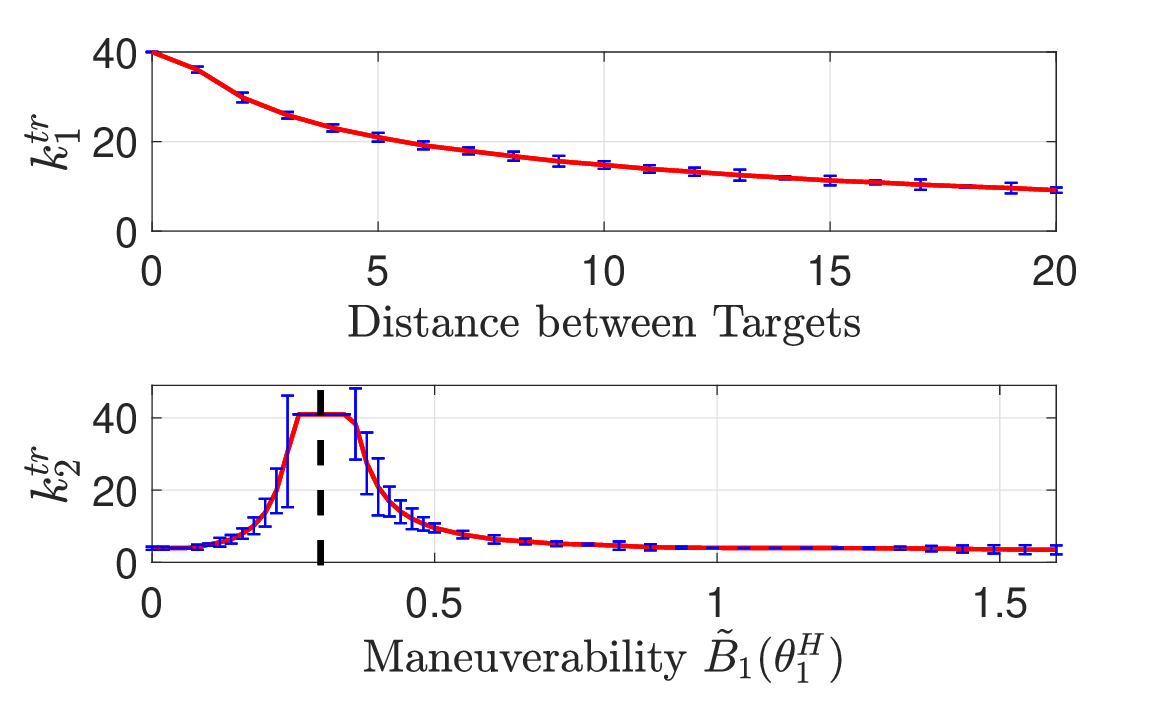}
\caption{ 
The plot of the deceived robot's truth-revealing stage versus the deceiver's type distinguishability. Error bars represent their variances, which are magnified by $5$ times. 
}
\label{fig:deceivability}
\end{figure}
When the maneuverability difference is negligible $\tilde{B}_1(\theta_1^H)\in (0.26, 0.36)$, the pursuer's type cannot be learned correctly in $K$ stages; i.e., the pursuer is $(K+1)$-stage $0$-deceivable. 
When the maneuverability difference is small, i.e., $\tilde{B}_1(\theta_1^H)\in (0.1, 0.5)$, yet not negligible, i.e., $\tilde{B}_1(\theta_1^H)\notin (0.26, 0.36)$, the variance of $k_2^{tr}$ is large.

Let $\theta_2=\theta_2^b$ be \textit{common knowledge} and assume that the evader's belief confirms to the prior distribution of the pursuer's type for all stages, i.e., $l_2^k(\theta_1|h^k,\theta^b)=\Xi_1(\theta_1), \forall \theta_1\in \Theta_1, \forall k\in \mathcal{K}$. Then, Fig. \ref{fig:newPoD} illustrates how the prior distribution of the pursuer's type \textcolor{black}{affects} the value of PoD under three scenarios: 
\begin{itemize}
    \item $\eta_1=1$, i.e., the central planner only evaluates UAV $1$'s performance under deception. 
    \item $\eta_1=0$, i.e., the central planner only evaluates UAV $2$'s performance under deception. 
    \item $\eta_1=0.5$, i.e., the central planner evaluates the average performance of two UAVs under deception. 
\end{itemize}
When the pursuer's type is also \textit{common knowledge}, i.e., 
$\Xi_1(\theta_1^H)=0$ (i.e., the pursuer has type $\theta_1^L$) and $\Xi_1(\theta_1^H)=1$ (i.e., the pursuer has type $\theta_1^H$), the game is of complete information and the value of PoD equals $1$. 
\textcolor{black}{Since PoD takes continuous values over  $\Xi_1(\theta_1^H)\in [0,1]$ and has a value of $1$ at two endpoints for all feasible $\eta_1$, we refer to the plots in Fig. \ref{fig:newPoD} as \textit{jump rope} plots.} 
\begin{figure}
\centering
\includegraphics[width=1 \columnwidth]{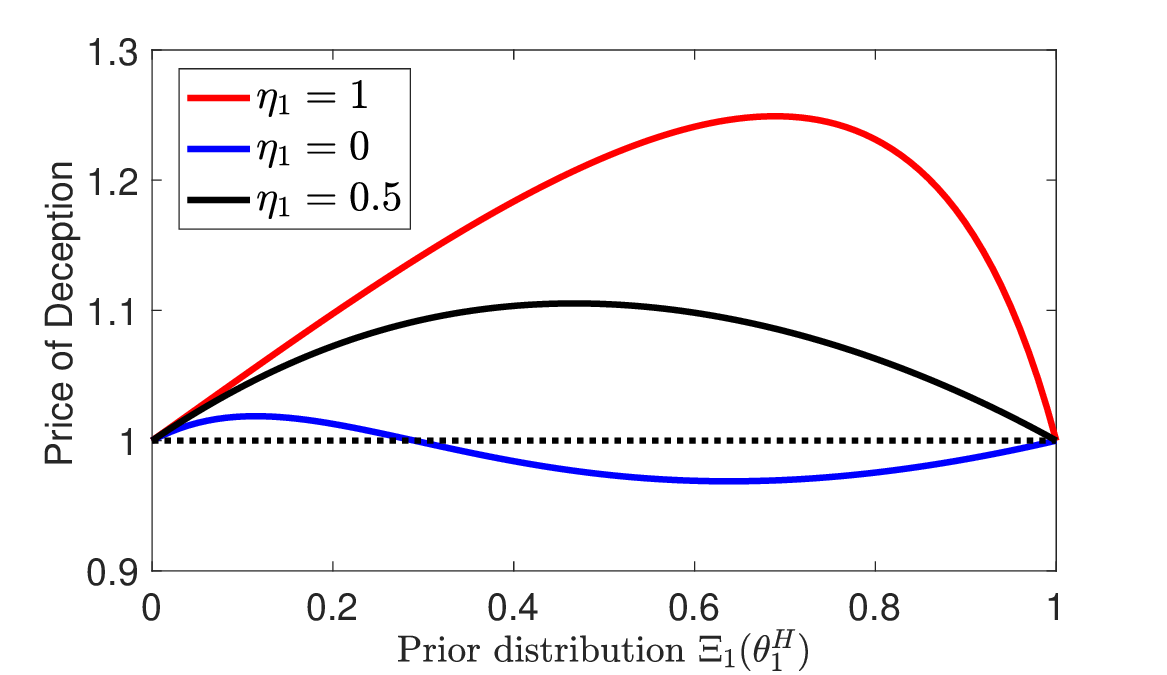}
\caption{ 
PoD vs. prior type distribution for three values of $\eta_1$. 
}
\label{fig:newPoD}
\end{figure}
They corroborate that the PoD can be bigger than $1$; i.e., deception among players may not only benefit the deceiver but also the deceivee.

\section{Conclusion and Future Work}
\label{sec:conclusion}
We have investigated a novel class of \textit{rational robot deception} problems where intelligent robots hide their heterogeneous private information to achieve their objectives  in finite stages with minimum costs. 
We have proposed an $N$-player dynamic game framework to quantify the impact of deception and design long-term optimal actions for deception and counter-deception. 
Robots form their own initial beliefs on others' private information and {update their beliefs at each stage based on extrinsic or intrinsic information.} 
Satisfying the properties of \textit{sequential rationality} and \textit{belief consistency}, perfect Bayesian Nash equilibrium can be used to predict $N$ robots' actions and costs over the $K$ stages. 
We have studied a class of games in the linear-quadratic form with {extrinsic belief dynamics}
to obtain a unique affine state-feedback control policy and a set of extended Riccati equations. 
The \textcolor{black}{cognitive} coupling resulted from the deception of types demonstrates a distinct feature of rational deception where each player's action hinges on not only his own belief but also all other players' beliefs. 
The concepts of \textit{deceivability}, \textit{distinguishability}, and \textit{reachability} have been defined to characterize the fundamental limits of deception.  \textcolor{black}{Meanwhile}, the \textit{price of deception} serves as a crucial evaluation and design metric. 

We have investigated a target protection problem where the evader aims to  deceptively reach the true target and the pursuer keeps her maneuverability as private information. 
The pursuer achieves a lower ex-post cumulative cost under the proposed policy  than under the direct-following and conservative policies. 
We have proposed multi-dimensional metrics such as the stage of  truth revelation and the endpoint distance to measure the deception impact throughout stages. We have concluded that Bayesian learning can largely reduce the impact of initial belief manipulation and sometimes result in a win-win situation.  The increase of the pursuer's maneuverability can also reduce the endpoint distance and her ex-post cumulative cost yet has a marginal effect. 
A robot is more \textit{deceivable}, i.e., less \textit{learnable} when its potential type is less \textit{distinguishable}. 
Finally, we have found that introducing additional deception to counteract existing deception is not always effective. 
Moreover, deception among multiple players may not only benefit the deceiver but also the deceivee.  




%



\ifCLASSOPTIONcaptionsoff
  \newpage
\fi


\bibliographystyle{IEEEtran}
\bibliography{IEEEabrv,RobotDecption}

\begin{IEEEbiography}[{\includegraphics[width=1in,height=1.25in,clip,keepaspectratio]{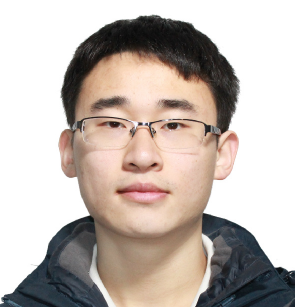}}]{Linan Huang}
(S'16) received the B.Eng. degree (Hons.) in Electrical Engineering from Beijing Institute of Technology, China, in 2016. He is currently pursuing the Ph.D. degree with the Laboratory for Agile and Resilient Complex Systems, Tandon School of Engineering, New York University, NY, USA. His research interests include dynamic decision making of the multi-agent system, mechanism design, artificial intelligence, security and resilience for the cyber-physical systems. 
\end{IEEEbiography}
\begin{IEEEbiography}[{\includegraphics[width=1in,height=1.25in,clip,keepaspectratio]{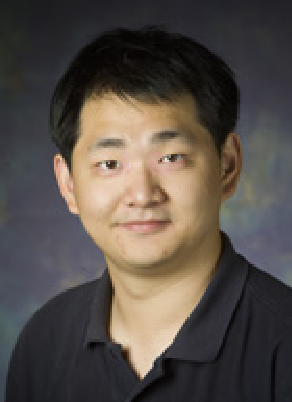}}]{Quanyan Zhu}
(SM’02-M’14) received B. Eng. in Honors Electrical Engineering from McGill University in 2006, M. A. Sc. from the University of Toronto in 2008, and Ph.D. from the University of Illinois at Urbana-Champaign (UIUC) in 2013. After stints at Princeton University, he is currently an associate professor at the Department of Electrical and Computer Engineering, New York University (NYU). He is an affiliated faculty member of the Center for Urban Science and Progress (CUSP) and Center for Cyber Security (CCS) at NYU. His current research interests include game theory, machine learning, cyber deception, and cyber-physical systems.
\end{IEEEbiography}

\end{document}